\documentclass[final]{siamonline1116}

\usepackage{braket,amsfonts}

\usepackage{array}

\usepackage[caption=false]{subfig}
\usepackage{pgfplots}

\newsiamthm{claim}{Claim}
\newsiamremark{rem}{Remark}
\newsiamremark{expl}{Example}
\newsiamremark{hypothesis}{Hypothesis}
\crefname{hypothesis}{Hypothesis}{Hypotheses}

\usepackage{algorithmic}

\usepackage{graphicx,epstopdf}

\Crefname{ALC@unique}{Line}{Lines}

\usepackage{amsopn}

\numberwithin{theorem}{section}

\usepackage{xspace}
\usepackage{bold-extra}
\usepackage[most]{tcolorbox}

\colorlet{texcscolor}{blue!50!black}
\colorlet{texemcolor}{red!70!black}
\colorlet{texpreamble}{red!70!black}
\colorlet{codebackground}{black!25!white!25}

\lstdefinestyle{siamlatex}{%
  style=tcblatex,
  texcsstyle=*\color{texcscolor},
  texcsstyle=[2]\color{texemcolor},
  keywordstyle=[2]\color{texemcolor},
  moretexcs={cref,Cref,maketitle,mathcal,text,headers,email,url},
}

\tcbset{%
  colframe=black!75!white!75,
  coltitle=white,
  colback=codebackground, 
  colbacklower=white, 
  fonttitle=\bfseries,
  arc=0pt,outer arc=0pt,
  top=1pt,bottom=1pt,left=1mm,right=1mm,middle=1mm,boxsep=1mm,
  leftrule=0.3mm,rightrule=0.3mm,toprule=0.3mm,bottomrule=0.3mm,
  listing options={style=siamlatex}
}

\newtcblisting[use counter=example]{example}[2][]{%
  title={Example~\thetcbcounter: #2},#1}

\newtcbinputlisting[use counter=example]{\examplefile}[3][]{%
  title={Example~\thetcbcounter: #2},listing file={#3},#1}

\DeclareTotalTCBox{\code}{ v O{} }
{ 
  fontupper=\ttfamily\color{black},
  nobeforeafter,
  tcbox raise base,
  colback=codebackground,colframe=white,
  top=0pt,bottom=0pt,left=0mm,right=0mm,
  leftrule=0pt,rightrule=0pt,toprule=0mm,bottomrule=0mm,
  boxsep=0.5mm,
  #2}{#1}

\usepackage{enumerate}

\patchcmd\newpage{\vfil}{}{}{}
\begin{tcbverbatimwrite}{tmp_\jobname_header.tex}
\title{Adjustments to Computer Models via Projected Kernel Calibration%
  \thanks{Submitted to the editors May 3, 2017.
\funding{This work is supported by NSF grant DMS 1564438 and also by the National Center for Mathematics and Interdisciplinary Sciences in CAS	and NSFC grants 11501551, 11271355 and 11671386.}}}

\author{Rui Tuo%
  \thanks{School of Industrial and Systems Engineering, Georgia Institute of Technology, GA 30309 and Academy of Mathematics and Systems Science, Chinese Academy of Sciences, Beijing, China 100190 (\email{tuorui@amss.ac.cn}).}%
}

\headers{Projected Kernel Calibration}
{Rui Tuo}
\end{tcbverbatimwrite}
\input{tmp_\jobname_header.tex}

\ifpdf
\hypersetup{ pdftitle={Adjustments to Computer Models via Projected Kernel Calibration} }
\fi

\begin{document}
\maketitle

\begin{tcbverbatimwrite}{tmp_\jobname_abstract.tex}
\begin{abstract}
Identification of model parameters in computer simulations is an important topic in computer experiments.
We propose a new method, called the projected kernel calibration method, to estimate these model parameters. The proposed method is proven to be asymptotic normal and semi-parametric efficient. As a frequentist method, the proposed method is as efficient as the $L_2$ calibration method proposed by Tuo and Wu [Ann. Statist. 43 (2015) 2331-2352]. On the other hand, the proposed method has a natural Bayesian version, which the $L_2$ method does not have. This Bayesian version allows users to calculate the credible region of the calibration parameters without using a large sample approximation. We also show that, the inconsistency problem of the calibration method proposed by Kennedy and
O’Hagan [J. R. Stat. Soc. Ser. B. Stat. Methodol. 63 (2001) 425-464] can be rectified by a simple modification of the kernel matrix.
\end{abstract}


\end{tcbverbatimwrite}
\input{tmp_\jobname_abstract.tex}

\section{Introduction}

With the development of mathematical modeling and computational techniques, there has become a wide spread use of computer simulations to study physical processes which can be expensive to observe or experiment with. An important task in computer modeling is to identify the model parameters involved in the computer code. For example, many computer simulators are built based on physical equations such as conservation laws. In these equations, there are constants or parameters, such as physical constants or inherent attributes of the physical objects, which cannot be controlled during the physical processes. Sometimes the values of these parameters are not known or cannot be measure directly from available physical experiments, and thus they need to be estimated using other physical observations. The activity of adjusting these model parameters are known as \textit{calibration} of the computer model, and these parameters are call \textit{calibration parameters}.

Kennedy and O'Hagan \cite{kennedy2001bayesian} propose a Bayesian framework for the calibration of computer models. They point out that there is a distance between the computer outputs and the physical responses, because the computer models are normally built under assumptions and simplifications which do not strictly hold true in reality. They also suggest to take this distance into account in the statistical model and a Gaussian process model is used to fit this discrepancy function. The Kennedy-O'Hagan (abbreviated as KO thereafter) model has been widely used in the literature in many disciplines. See \cite{higdon2004combining,goldstein2004probabilistic,higdon2008computer,bayarri2007framework,joseph2014engineering,gramacy2015calibrating}, among others. However, because of the use of the discrepancy function, the method has an identifiability issue. It is shown in \cite{tuo2016calibration} that, as the sample size goes to infinity, the limit value of the KO parameter estimation depends on the choice of the prior and thus it should be considered as inconsistent.

A mathematical framework for calibration is suggested by \cite{tuo2014efficient,tuo2016calibration}. a new calibration method, called the $L_2$ calibration is introduced in \cite{tuo2014efficient}, which is proven to enjoy nice frequentist properties including consistency, asymptotic normality and semi-parametric efficiency. However, the $L_2$ calibration only provides a point estimate of the calibration parameter, which do not meet the practical needs of \textit{uncertainty quantification}. Recall that we are also interested in the uncertainty of the estimate. In practice, users usually prefer Bayesian methods because the physical sample size is usually rather small and Bayesian methods are more flexible in small sample size. Because $L_2$ calibration is a two-step approach, it does not admit a simple Bayesian version.

In this work, we propose a novel calibration approach called the \textit{projected kernel calibration}. The main idea is to use an orthogonality condition arisen in the framework of \cite{tuo2014efficient}. In view of this orthogonality condition, we project the kernel function onto some linear subspace and use the projected kernel to build basis functions to construct non-parametric estimators. The projected kernel calibration is proven to be asymptotically equivalent to the $L_2$ calibration, and it also has a natural Bayesian interpretation. The Bayesian version of the proposed method has a similar expression as the method by \cite{kennedy2001bayesian}. This implies that we can rectify the inconsistency problem of KO method by a simple modification.

The reminder of this article is organized as follows. In Section \ref{sec_background}, we state the aim of calibration for computer models and review the framework suggested by \cite{tuo2014efficient,tuo2016calibration}. In Section \ref{sec_KernelsProperties}, we introduce the projected kernel functions and study their properties. We propose a new type of calibration method in Section \ref{sec_PKC}. The asymptotic properties and Bayesian interpretation are also given in this section. Concluding remarks are given in Section \ref{sec_discussion}. Technical proofs are given in Section \ref{sec_proofs}.



%
%
%

\section{Background}\label{sec_background}

The calibration of computer models is the activity of adjusting the model parameters of a computer simulator so that the outputs of the computer code fit the physical responses.

Denote the experimental region of the physical experiments by $\Omega$, which is assumed to be a convex and compact subregion of $\mathbb{R}^d$. Suppose we have a sequence of physical observations, denoted as $(x_i,y_i), i=1,2,\ldots, n$. We follows a standard modeling assumption imposed in the computer experiments literature like \cite{kennedy2001bayesian} and use the following nonparametric model to link the physical response $y_i$ and the covariant $x_i$:
\begin{eqnarray}
y_i=\zeta(x_i)+e_i,
\end{eqnarray}
where $\zeta$ is an unknown continuous function over $\Omega$, referred to as the \textit{true process}; $e_i$'s are independent and identically distributed random variables with mean zero and finite variance. For simplicity, we suppose the design points $\{x_i\}$ is a sequence of independent and identically distributed random variables following the uniform distribution over $\Omega$.

\subsection{Purpose of Calibration}

Let $y^s(x,\theta)$ be the output of the \textit{deterministic} computer simulator given the control variable $x$ and the calibration parameter $\theta$. Denote the parameter space of $\theta$ by $\Theta$. Suppose $\Theta$ is a compact subset of $\mathbb{R}^q$. As discussed in \cite{tuo2014efficient,tuo2016calibration}, there are two types of computer simulators. In the first case, the cost of running the computer code is negligible so that we can treat $y^s$ simply as a known function and we refer such computer codes as ``cheap'' codes. In the second case, each run of the computer code is costly so that only a limited number of runs can be conducted. Such computer codes are said to be ``expensive''. For the expensive code, a standard approach in computer experiments is to run the computer code over a selected set of scattered points and then build a surrogate model, denoted as $\hat{y}^s$. Although the computer codes are said to be expensive, they are much less costly than the corresponding physical runs. Therefore, the number of computer runs should be larger than the physical sample size. Moreover, the computer code is deterministic while the physical responses are noisy. As a consequence, it is reasonable to believe that the approximation error of the surrogate modeling is much smaller than the statistical estimation error for the true process $\zeta$. For convenience, in this work we only consider the case where the computer code is cheap. For expensive computer codes, the main results of this work is still valid provided that $\|y^s-\hat{y}^s\|_{L_\infty(\Omega)}$ and $\|\frac{\partial y^s}{\partial\theta}-\frac{\partial \hat{y}^s}{\partial\theta}\|_{L_\infty(\Omega)}$ are sufficiently small. We refer to \cite{tuo2014efficient} for detailed conditions and arguments.

Calibration for computer models, first suggested by \cite{kennedy2001bayesian}, is to find a good value for the vector of calibration parameters so that the computer response surface ``matches'' the physical observations best. Kennedy and O'Hagan \cite{kennedy2001bayesian} do not give a mathematical definition of the optimal values of the calibration parameters. In \cite{tuo2014efficient}, the optimal choice of the calibration parameters is defined by the $L_2$ projection
\begin{eqnarray}
\theta^*:=\operatorname*{argmin}_{\theta\in\Theta}\|\zeta(\cdot)-y^s(\cdot,\theta)\|_{L_2(\Omega)}.\label{l2projection}
\end{eqnarray}

\subsection{Review on $L_2$ Calibration}\label{sec_l2}

Kennedy and O'Hagan \cite{kennedy2001bayesian} propose a Bayesian approach to estimate the calibration parameter. However, it is known that the KO model suffers from some identifiability issue. As a consequence of this identifiability problem, the KO method can give highly unstable answers which rely heavily on the prior distributions. See \cite{plumlee2016bayesian} for more discussions. From a frequentist point of view, the estimator is inconsistent because of the non-identifiability. In work by \cite{tuo2016calibration} studies the asymptotic behavior of the KO approach and point out that it may render unreasonable results.

The $L_2$ calibration method is proposed by \cite{tuo2014efficient}, which proceeds in two steps. First, estimate the true process in a nonparametric manner by solving the following smoothing problem in the native space
\begin{eqnarray}\label{L2step1}
\hat{\zeta}_n=\operatorname*{argmin}_{f\in\mathcal{N}_\Phi(\Omega)}\frac{1}{n}\sum_{i=1}^n(y^p_i-f(x_i))^2+ \lambda\|f\|^2_{\mathcal{N}_\Phi(\Omega)}.
\end{eqnarray}
A detailed discussion about the native space $\mathcal{N}_\Phi(\Omega)$ will be given in Section \ref{sec_nativespace}. Next, calculate the $L_2$ calibration estimate defined as
\begin{eqnarray}\label{L2step2}
\hat{\theta}^{L_2}=\operatorname*{argmin}_{\theta\in\Theta}\|\zeta(\cdot)-y^s(\cdot,\theta)\|_{L_2(\Omega)}.
\end{eqnarray}

A remarkable theoretical property of the $L_2$ calibration is the \textit{semi-parametric efficiency}. \cite{tuo2014efficient} prove that under certain conditions the $L_2$ calibration estimator has the asymptotic representation
\begin{eqnarray}\label{benchmark}
\hat{\theta}_n^{L_2}-\theta^*=-2 V^{-1}\left\{\frac{1}{n}\sum_{i=1}^n e_i \frac{\partial y^s}{\partial \theta}(x_i,\theta^*)\right\}+o_p(n^{-1/2}),
\end{eqnarray}
with
\begin{eqnarray*}
	V=E\left[\frac{\partial^2}{\partial\theta\partial\theta^T}(\zeta(x_i)-y^s(x_i,\theta^*))^2\right],
\end{eqnarray*}
and thus is asymptotically normally distributed.
Moreover, the $L_2$ calibration is semi-parametric efficient, in the sense that there does not exist a regular estimator with an even smaller asymptotic variance.

The semi-parametric efficiency is a benchmark for new calibration estimators to be found later.

\subsection{Goal of This Work}

In view of the semi-parametric efficiency property, the $L_2$ calibration is not improvable. However, there remains some drawbacks of the $L_2$ calibration method. In practice, researchers and engineers wish not only to obtain a point estimate for the calibration parameters, but also a confidence/credible region, in order to assess the uncertainty of the estimator. Although the asymptotic normality of the $L_2$ calibration admits a natural asymptotic confidence ellipsoid, in reality it is usually too rough to use the asymptotic distribution of the $L_2$ calibration estimator directly because the physical sample size is typically rather small. Furthermore, its performance is deteriorated by the estimation of the unknown parameters in the asymptotic distribution. As stated before, the high cost of physical experiments prohibits a large number of physical runs and this is why we seek the help of computer experiments.

Because small sample problems are common in expensive physical or computer experiments, researchers and engineers favor Bayesian approaches which are more flexible to sample size. For the calibration of computer models, the KO-type estimators \cite{kennedy2001bayesian,higdon2004combining} are of this kind. However, \cite{tuo2016calibration} conduct some asymptotic analysis and show that these estimators may render unreasonable results for calibration. 

In order to find a stable Bayesian approach for calibration, one may want to construct a Bayesian version of the $L_2$ calibration. Unfortunately, it seems that there does not exist a simple Bayesian version or interpretation of the $L_2$ calibration. The main reason is that the $L_2$ calibration is a two-step approach, which does not allow us to construct an explicit expression of the posterior density.

The goal of this work is to find an estimator for the calibration parameters which enjoys two properties. First, the estimator is semi-parametric efficient, i.e., the estimator admits an asymptotic representation as shown in (\ref{benchmark}). Second, the estimator can be interpreted from a Bayesian point of view, i.e., the point estimator is the posterior mode of a Bayesian estimator.

\section{Projected Kernels and Their Properties}\label{sec_KernelsProperties}

We first briefly introduce the work of \cite{plumlee2016orthogonal} on Orthogonal Gaussian process models,
which inspires us to introduce the projected kernels.

\subsection{Review on Orthogonal Gaussian Process Models}\label{sec_ogp}

Let $\mathcal{G}$ be a finite-dimensional subspace of $L_2(\Omega)$ with $\dim \mathcal{G}=m$. For any $f\in L_2(\Omega)$, let $\mathcal{P}_{\mathcal{G}}f$ be the projection of $f$ onto $\mathcal{G}$, given by
\begin{eqnarray*}
	\mathcal{P}_{\mathcal{G}}f=\sum_{i=1}^m \langle f,e_i\rangle_{L_2(\Omega)} e_i,
\end{eqnarray*}
where $\{e_1,\dots,e_m\}$ forms an orthonormal basis of $\mathcal{G}$. Define the perpendicular component
\begin{eqnarray*}
	\mathcal{P}_{\mathcal{G}}^\perp f=f-\mathcal{P}_{\mathcal{G}}f.
\end{eqnarray*}

Let $Z(\cdot)$ be a Gaussian process over $\Omega$ with mean zero and covariance function $K(\cdot,\cdot)$. It can be shown that $\mathcal{P}_{\mathcal{G}}^\perp Z(\cdot)$ is also a Gaussian process with mean zero, and its covariance function is
\begin{eqnarray}
&&K_\mathcal{G}(s,t)=Cov\left(Z(s)-\sum_{i=1}^m \langle Z,e_i\rangle_{L_2(\Omega)} e_i(s), Z(t)-\sum_{i=1}^m \langle Z,e_i\rangle_{L_2(\Omega)} e_i(t)\right)\nonumber\\
&&=K(s,t)-\sum_{i=1}^m e_i(s)\int_\Omega K(x,t)e_i(x) d x
-\sum_{j=1}^m e_j(t)\int_\Omega K(s,y)e_j(y) d y\nonumber\\
&&\text{~~~~}+\sum_{i=1}^m\sum_{j=1}^m e_i(s)e_j(t)\int_{\Omega\times\Omega} K(x,y)e_i(x)e_j(y) d x d y,\label{Kg}
\end{eqnarray}
for $s,t\in\Omega$.
The process $\mathcal{P}_{\mathcal{G}}^\perp Z$ is referred to as the orthogonal Gaussian process of $Z$ with respect to $\mathcal{G}$ in \cite{plumlee2016orthogonal}.


\subsection{Projected Kernels}

Although $K_\mathcal{G}$ in (\ref{Kg}) is derived from a Gaussian process, we can regard (\ref{Kg}) as the definition of $K_\mathcal{G}^\Omega$ without a Gaussian process context. For the ease of the mathematical treatments later, we introduce some linear operators. 


Given a measurable set $\Omega\subset\mathbb{R}^d$, a finite dimensional subspace $\mathcal{G}\subset L_2(\Omega)$, define the linear operators $\mathcal{P}_\mathcal{G}^{(1)},\mathcal{P}_\mathcal{G}^{(2)}:L_2(\Omega\times\Omega)\rightarrow L_2(\Omega\times\Omega)$ as
\begin{eqnarray*}
	(\mathcal{P}_\mathcal{G}^{(1)} u)(x,y)&=&\sum_{i=1}^m e_i(x)\int_\Omega u(s,y) e_i(s) d s,\\
	(\mathcal{P}_\mathcal{G}^{(2)} u)(x,y)&=&\sum_{i=1}^m e_i(y)\int_\Omega u(x,t) e_i(t) d t,
\end{eqnarray*}
for $u\in L_2(\Omega\times\Omega), x,y\in\Omega$,
where $\{e_i\}_{i=1}^m$ forms an orthonormal basis of $\mathcal{G}$. It is easily verified that the definitions of $\mathcal{P}_\mathcal{G}^{(1)}$ and $\mathcal{P}_\mathcal{G}^{(2)}$ are independent of the choice of $\{e_i\}$. 
Define the bilinear operator $\kappa: L_2(\Omega\times\Omega)\times L_2(\Omega)\rightarrow L_2(\Omega)$ as
\begin{eqnarray*}
	\kappa(u,f)&=&\int_\Omega u(\cdot,y) f(y) d y,
\end{eqnarray*}
for $u\in L_2(\Omega\times\Omega),f\in L_2(\Omega)$.

Now we can define the projected kernel.

\begin{definition}
	Suppose $\Omega$ is a measurable subset of $\mathbb{R}^d$, $K(\cdot,\cdot)$ is a positive definite function over $\Omega\times\Omega$, $\mathcal{G}$ is a finite dimensional subspace of $L_2(\Omega)$. Define the projected kernel of $K$ with respect to $\mathcal{G}$ by
	\begin{eqnarray}
	K_\mathcal{G}=K-\mathcal{P}_\mathcal{G}^{(1)}K-\mathcal{P}_\mathcal{G}^{(2)}K+\mathcal{P}_\mathcal{G}^{(1)}\mathcal{P}_\mathcal{G}^{(2)}K. \label{Def1}
	\end{eqnarray}
\end{definition}


By the symmetricity of $K$, it can be verified that
\begin{eqnarray*}
	(\mathcal{P}_\mathcal{G}^{(1)}K)(x,y)&=&(\mathcal{P}_\mathcal{G}^{(2)}K)(y,x), \\
	(\mathcal{P}_\mathcal{G}^{(1)}\mathcal{P}_\mathcal{G}^{(2)}K)(x,y)&=&(\mathcal{P}_\mathcal{G}^{(1)}\mathcal{P}_\mathcal{G}^{(2)}K)(y,x).
\end{eqnarray*}
Therefore, $K_\mathcal{G}$ in (\ref{Def1}) is also a symmetric function.
It is easily verified that the projected kernel $K_\mathcal{G}$ can be expressed as the right hand side of (\ref{Kg}) if $\{e_i\}_{i=1}^m$ is an orthonormal basis of $\mathcal{G}$.
Let $\mathcal{G}^\perp$ be the orthogonal complement of $\mathcal{G}$ in $L_2(\Omega)$.




 It is proved in \cite{plumlee2016orthogonal} that $K_\mathcal{G}$ defined in (\ref{Def1}) is semi-positive definite under mild conditions. However, here we have to pursue a stronger result. Our goal is to show that $K_\mathcal{G}$ is in fact positive definite. The positive definiteness is crucial in both orthogonal Gaussian process modeling \cite{plumlee2016orthogonal} and the interpolation scheme to be discussed later in this work. In Gaussian process modeling, the positive definiteness of the covariance function is necessary to ensure that the covariance matrix is invertible for all possible design sets (with distinct points), which allows for finite likelihood values. In the interpolation scheme discussed later, we also need $K_\mathcal{G}$ to be positive definite to guarantee the uniqueness of the interpolant.

Now we give some sufficient conditions that ensure the positive definiteness of $K_\mathcal{G}$. Using the terminology of Gaussian process modeling, we assume that the Gaussian process $Z(\cdot)$ discussed in Section \ref{sec_ogp} is stationary, i.e., there exists $R(\cdot)$, such that $K(x,y)=R(x-y)$ for all $x,y\in\Omega$. Besides, we assume that $R(x)$ can be defined for all $x\in\mathbb{R}^d$ and $H(x,y)=R(x-y)$ is positive definite over $\mathbb{R}^d\times\mathbb{R}^d$.
In Gaussian process modeling and radial basis functions interpolation, stationary covariance functions over $\mathbb{R}^d$ like the power exponential and the Mat\'{e}rn families are commonly used \cite{santner2003design,wendland2005scattered}. Theorem \ref{Thpositivedefinite} shows that a large class of stationary kernels can produce positive definite projected kernels.

\begin{theorem}\label{Thpositivedefinite}
	Suppose $R\in\L_1(\mathbb{R}^d)$ possesses a real-valued Fourier transform satisfying $\tilde{R}>0$ almost everywhere. Define $K_\mathcal{G}^\Omega$ by (\ref{Kg}) with $K(x,y)=R(x-y)$. Then $K_\mathcal{G}$ is positive definite for all $\Omega\subset \mathbb{R}^d$ with Lebesgue measure $\mathfrak{m}(\Omega)\in(0,+\infty)$ and all finite-dimensional $\mathcal{G}\in L_2(\Omega)$.
\end{theorem}

We note that the Fourier transform of the correlation function of a stationary process is referred to as its spectral density. The spectral densities of commonly used stationary Gaussian processes can be found in books such as \cite{stein1999interpolation} or \cite{santner2003design}.
By Theorem \ref{Thpositivedefinite}, for many common correlation functions like the Gaussian and the Mat\'{e}rn families, the covariance functions for the orthogonal processes are positive definite under a broad choice of $\Omega$ and $\mathcal{G}$.

\subsection{Native Spaces}\label{sec_nativespace}

In this section we investigate the native spaces (also referred to as the reproducing kernel Hilbert spaces) generated by $K$ and $K_\mathcal{G}^\Omega$ in (\ref{Kg}).

Given a symmetric and positive definite continuous function $\Phi$ over $\Omega\times\Omega$, define the linear space
\begin{eqnarray*}
	F_\Phi(\Omega)=\left\{\sum_{i=1}^N \beta_i\Phi(\cdot,x_i):N\in\mathbb{N},\beta_i\in\mathbb{R},x_i\in \Omega\right\}
\end{eqnarray*}
and equip this space with the bilinear form
\begin{eqnarray*}
	\Big\langle\sum_{i=1}^N \beta_i\Phi(\cdot,x_i),\sum_{j=1}^M \gamma_j\Phi(\cdot,y_j)\Big\rangle_\Phi:=\sum_{i=1}^N\sum_{j=1}^M\beta_i\gamma_j\Phi(x_i,y_j).
\end{eqnarray*}
Define the native space $\mathcal{N}_\Phi(\Omega)$ as the closure of $F_\Phi(\Omega)$ under the inner product $\langle\cdot,\cdot\rangle_\Phi$. The inner product of $\mathcal{N}_\Phi(\Omega)$, denoted as $\langle\cdot,\cdot\rangle_{\mathcal{N}_\Phi(\Omega)}$, is induced by $\langle\cdot,\cdot\rangle_\Phi$. Define the native norm as $\|f\|_{N_\Phi(\Omega)}=\sqrt{\langle f,f\rangle_{\mathcal{N}_\Phi(\Omega)}}$. We refer to \cite{wendland2005scattered} for more discussions about the native spaces and their properties.

Because $K$ is positive definite over $\Omega\times\Omega$, it generates native space $\mathcal{N}_K(\Omega)$. Theorem \ref{Thpositivedefinite} guarantees the positive definiteness of many commonly encountered projected kernels. From now on, we always assume that the projected kernel $K_\mathcal{G}$ is positive definite so that it generates the native space $\mathcal{N}_{K_\mathcal{G}}(\Omega)$.

Since $K_\mathcal{G}$ is induced from $K$, it is of interest to explore certain relationship between $\mathcal{N}_K(\Omega)$ and $\mathcal{N}_{K_\mathcal{G}}(\Omega)$. Because the definition of $K_\mathcal{G}$ also involves projection in $L_2(\Omega)$, it is natural to start with embedding the native spaces into $L_2(\Omega)$. Such a result is given by Lemma \ref{th_embedding} in Section \ref{sec_proofs}, which collects the results from Lemma 10.27 and Proposition 10.28 of \cite{wendland2005scattered}.

In practice, it is reasonable to assume $\mathcal{G}\subset\mathcal{N}_K(\Omega)$, because the space $\mathcal{N}_K(\Omega)$ is regarded as the set of ``all possible functions of interest'' and thus should contain all ``regular'' functions like the elements of $\mathcal{G}$. Besides, we will show that the native space $\mathcal{N}_{K_\mathcal{G}}(\Omega)$ enjoys some nice properties by assuming $\mathcal{G}\subset\mathcal{N}_K(\Omega)$.

Intuitively, we expect that $\mathcal{N}_{K_\mathcal{G}}(\Omega)$ is formed by certain ``projected functions'' from $\mathcal{N}_K(\Omega)$. A natural question is whether we can bound $\|\mathcal{P}_\mathcal{G}^\perp f\|_{\mathcal{N}_{K_\mathcal{G}}(\Omega)}$ by a multiple of $\|f\|_{\mathcal{N}_K(\Omega)}$. Such a result is given in Theorem \ref{th_norminequality}, which is an important step to establish the asymptotic theory in Section \ref{sec_asymptotic}.

\begin{theorem}\label{th_norminequality}
	Suppose $\Omega\subset \mathbb{R}^d$ is compact and $K$ is a symmetric positive definite function over $\Omega\times\Omega$. Let $\mathcal{G}\subset \mathcal{N}_K(\Omega)$ be finite dimensional. Then the following statements are true:
	\begin{enumerate}[(i)]
		\item For any $f\in\mathcal{N}_K(\Omega)$, we have $\mathcal{P}_\mathcal{G}^\perp f\in\mathcal{N}_{K_\mathcal{G}}(\Omega)$ with
		\begin{eqnarray}
		\|\mathcal{P}_\mathcal{G}^\perp f\|_{\mathcal{N}_{K_\mathcal{G}}(\Omega)}\leq C_1 \|f\|_{\mathcal{N}_K(\Omega)}, \label{normineq}
		\end{eqnarray}
		where
		\begin{eqnarray*}
			C_1=1+\sup_{\substack{g\in\mathcal{G}\\ \|g\|_{L_2(\Omega)}=1}}\|g\|_{\mathcal{N}_K(\Omega)}\|\kappa(K,g)\|_{\mathcal{N}_K(\Omega)}.
		\end{eqnarray*}
		\item For any $f\in\mathcal{N}_{K_\mathcal{G}}(\Omega)$, we have $f\in\mathcal{N}_K(\Omega)$ with
		\begin{eqnarray}
		\|f\|_{\mathcal{N}_K(\Omega)}\leq C_2 \|f\|_{\mathcal{N}_{K_\mathcal{G}}(\Omega)},
		\end{eqnarray}
		where
		\begin{eqnarray*}
			C_2=1+\sup_{\substack{g\in\mathcal{G}\\ \|g\|_{L_2(\Omega)}=1}}\|g\|_{\mathcal{N}_K(\Omega)}\left(\int_\Omega K(x,x)d x\right)^{1/2},\\
		\end{eqnarray*}
	\end{enumerate}
\end{theorem}

\begin{rem}
	Given $K$ and $\mathcal{G}$, $C_1$ and $C_2$ are finite because $\mathcal{G}$ is a finite dimensional space and all norms over a finite dimensional space are equivalent.
\end{rem}

\begin{rem}
	Using Lemma \ref{Thmercer} in Section \ref{sec_lemma} it is easily verified that
	\begin{eqnarray}
	\|\kappa(K,g)\|_{\mathcal{N}_K(\Omega)}\leq \rho_{max}\|g\|_{\mathcal{N}_K(\Omega)}, \label{rhomax}
	\end{eqnarray}
	where $\rho_{max}$ is the greatest eigenvalue of $\kappa(K,\cdot)$.
\end{rem}

Corollary \ref{coro_directsum} is an immediate consequence of Theorem \ref{th_norminequality} together with Lemma \ref{th_embedding} in Section \ref{sec_lemma}.

\begin{corollary}\label{coro_directsum}
	As linear spaces, 
	$\mathcal{N}_K(\Omega) = \mathcal{N}_{K_\mathcal{G}}(\Omega)\oplus\mathcal{G}$.
\end{corollary}

\section{Projected Kernel Calibration}\label{sec_PKC}
%
%
%

In this section we introduce the proposed projected kernel calibration method and discuss its properties.

\subsection{Methodology}

In Section \ref{sec_nativespace}, we introduce the native space $\mathcal{N}_{K_\mathcal{G}}(\Omega)$, which is a subset of $\mathcal{G}^\perp$, the $L_2(\Omega)$-orthogonal complement of $\mathcal{G}$. The orthogonality plays an important role in estimating the $L_2$ projection $\theta^*$ defined in (\ref{l2projection}). To see this, we suppose $\theta^*$ is an interior point of $\Theta$. Then, we differentiate the right hand side of (\ref{l2projection}) and use the optimality condition to obtain
\begin{eqnarray}
0 = \int_\Omega \frac{\partial y^s}{\partial \theta}(x,\theta^*)(\zeta(x)-y^s(x,\theta^*)) d x.\label{orthogonality}
\end{eqnarray}
For $\theta=(\theta_1,\ldots,\theta_q)\in\Theta\subset \mathbb{R}^q$, define
\begin{eqnarray}
\mathcal{G}_\theta = \operatorname{span}\left\{\frac{\partial y^s}{\partial\theta_i}(\cdot,\theta):i=1,\ldots,q\right\}.
\end{eqnarray}
Then (\ref{orthogonality}) implies that $\zeta(\cdot)-y^s(\cdot,\theta^*)$ is orthogonal to $\mathcal{G}_{\theta^*}$ in $L_2(\Omega)$.

As suggested by \cite{tuo2014efficient}, the goal of calibration is to estimate $\theta^*$. Let $\hat{\theta}$ be an estimator of $\theta^*$. In view of (\ref{orthogonality}), it is natural to impose the following analogous orthogonality requirement:
\begin{eqnarray}
(\hat{\zeta}(\cdot)-y^s(\cdot,\hat{\theta})) \perp \mathcal{G}_{\hat{\theta}}, \label{perp}
\end{eqnarray}
where $\hat{\zeta}$ is an estimate of $\zeta$. Recall that \cite{tuo2014efficient} suggest to estimate $\zeta$ via smoothing in a native space, denoted by $\mathcal{N}_K(\Omega)$. In this work, we use a similar idea and suppose that $\hat{\zeta}$ and $y^s(\cdot,\hat{\theta})$ lie in $\mathcal{N}_K(\Omega)$. Now in view of (\ref{perp}) and Corollary \ref{coro_directsum}, we find that $\hat{\delta}=\hat{\zeta}-y^s(\cdot,\hat{\theta})\in \mathcal{N}_{K_\mathcal{G}}(\Omega)$. This inspires us to introduce the projected kernel smoothing estimator $(\hat{\theta}_{PK},\hat{\delta}_{PK})$ as the minimizer of
\begin{eqnarray}
& &  \min_{\theta\in\Theta,\delta\in\mathcal{N}_{K_{\mathcal{G}_\theta}}(\Omega)} \frac{1}{n}\sum_{i=1}^n (y_i^p-\delta(x_i)-y^s(x_i,\theta))^2+\lambda\|\delta\|^2_{\mathcal{N}_{K_{\mathcal{G}_\theta}}(\Omega)}\nonumber \\
&=&  \min_{\theta\in\Theta}\min_{\delta\in\mathcal{N}_{K_{\mathcal{G}_\theta}}(\Omega)}\frac{1}{n}\sum_{i=1}^n (y_i^p-\delta(x_i)-y^s(x_i,\theta))^2+\lambda\|\delta\|^2_{\mathcal{N}_{K_{\mathcal{G}_\theta}}(\Omega)},\label{PKC}
\end{eqnarray}
where $\lambda$ is a tuning parameter. We suggest choosing $\lambda$ by generalized cross validation (GCV); see \cite{wahba1990spline}. Although (\ref{PKC}) is formulated as an infinite-dimensional minimization problem, the inner minimization problem can be solved analytically with the help of the representer theorem \cite{scholkopf2001generalized,wahba1990spline}; see Section \ref{sec_Bayesian} for details.

%

\subsection{Asymptotic Theory}\label{sec_asymptotic}

In this section, we investigate the asymptotic properties of the proposed projected kernel calibration estimators. We will show that under certain conditions these estimators attain the semi-parametric efficiency mentioned in Section \ref{sec_l2}. To be precise, in this section we use the notation $\hat{\theta}_n,\hat{\delta}_n$ instead of $\hat{\theta}_{PK},\hat{\delta}_{PK}$ defined in (\ref{PKC}), where $n$ is the physical sample size. The tuning parameter $\lambda$ in (\ref{PKC}) should also vary with $n$, denoted as $\lambda_n$. As suggested in \cite{tuo2014efficient}, we assume that the design points $x_i$'s are independent samples from the uniform distribution over $\Omega$.

In the first asymptotic result, we are concerned with the predictive power of the proposed method, in terms of estimating the true process $\zeta(\cdot)$. Let $\hat{\zeta}_n(\cdot)=\hat{\delta}_n(\cdot)+y^s(\cdot,\hat{\theta}_n)$, which is a natural estimator of $\zeta(\cdot)$. 

\begin{theorem}\label{th_prediction}
	Suppose $x_i$'s are independent random samples from the uniform distribution over $\Omega$ and $\mathcal{N}_K(\Omega)$ can be continuously embedded into the Sobolev space $H^m(\Omega)$ with $m>d/2$. There exists $C_0>0$ such that
	\begin{eqnarray}\label{subexponential}
	E[\exp(C_0|e_i|)]<\infty,
	\end{eqnarray}
	Moreover, we assume the following uniform boundedness conditions:
	\begin{eqnarray}
	C_3&:=&\sup_{\theta\in\Theta,i=1,\ldots,q}\left\{\left\|\frac{\partial y^s}{\partial \theta_i}(\cdot,\theta)\right\|_{\mathcal{N}_K(\Omega)}/\left\|\frac{\partial y^s}{\partial \theta_i}(\cdot,\theta)\right\|_{L_2(\Omega)}\right\}<\infty,\label{supbound}\\
	C'_3&:=&\sup_{\theta\in\Theta}\|y^s(\cdot,\theta)\|_{\mathcal{N}_K(\Omega)}<\infty.\label{supbound2}
	\end{eqnarray}
	Then if $\lambda_n \sim n^{-\frac{2 m}{2m+d}}$, we have
	\begin{eqnarray}
	\|\zeta-\hat{\zeta}_n\|_{L_2(\Omega)} &=& O_p(n^{-\frac{m}{2m+d}}),\label{L2rate} \\
	\|\hat{\zeta}_n\|_{\mathcal{N}_K(\Omega)} &=& O_p(1). \label{Hmrate}
	\end{eqnarray}
\end{theorem}

The rate of convergence in (\ref{L2rate}) is known to be optimal in the current context; see \cite{stone1982optimal}.
Next, we turn to the calibration consistency, i.e., whether $\hat{\theta}_n$ converges to the $L_2$ calibration $\theta^*$.

\begin{theorem}\label{th_calibration}
	In addition to the assumptions of Theorem \ref{th_prediction}, we assume 
	that the matrix
	\begin{eqnarray}\label{Istar}
	V = \int_\Omega \frac{\partial^2}{\partial \theta^T\partial\theta} (\zeta(x)-y^s(x,\theta^*))^2 d x
	\end{eqnarray}
	is positive definite. Moreover, there exists a neighborhood of $\theta^*$, denoted as $\Theta'$, satisfying
	\begin{eqnarray}\label{supbound3}
	\sup_{\theta\in\Theta',1\leq i\leq q}\left\|\frac{\partial y^s}{\partial \theta_i}(\cdot,\theta)\right\|_{\mathcal{N}_{K_{\mathcal{G}_\theta}}(\Omega)}\leq \infty.
	\end{eqnarray}
	Then there exists a local minimum point of (\ref{PKC}), denoted as $(\hat{\theta}^*_n,\hat{\delta}_n^*)$, such that the sequence $\{\hat{\theta}^*_n\}$ converges to $\theta^*$ in probability as $n$ goes to infinity.
\end{theorem}

Although the nonparametric estimator $\hat{\zeta}_n$ converges at a rate lower than $O(n^{-1/2})$ as in (\ref{L2rate}), it is shown by Theorem \ref{th_normality} that the consistent estimator $\hat{\theta}^*_n$ in Theorem \ref{th_calibration} has $O(n^{-1/2})$ rate of convergence. Define matrix
\begin{eqnarray*}
	D_\theta= \left(\left\langle\frac{\partial y^s(\cdot,\theta)}{\partial\theta_i},\frac{\partial y^s(\cdot,\theta)}{\partial\theta_j}\right\rangle_{L_2(\Omega)}\right)_{i j}.
\end{eqnarray*}
Let $\lambda_{min}(D_\theta)$ be the minimum eigenvalue of $D_\theta$.

\begin{theorem}\label{th_normality}
	Under the conditions of Theorem \ref{th_calibration}, let $\{(\hat{\theta}^*_n,\hat{\delta}_n^*)\}$ be a sequence of local minimum points such that $\hat{\theta}_n^*$ converges to $\theta^*$ in probability as shown in Theorem \ref{th_calibration}. In addition, we suppose
	\begin{eqnarray}
	&&\sup_{1\leq i,j\leq q,\theta\in\Theta'} \left\|\frac{\partial^2 y^s}{\partial \theta_i\partial\theta_j}(\cdot,\theta)\right\|_{\mathcal{N}_{K_{\mathcal{G}_\theta}}(\Omega)}\leq \infty, \label{supbound4} \\
	&&\inf_{\theta\in\Theta'}\lambda_{min}(D_\theta) > 0.\label{eigenbound}
	\end{eqnarray}
	Then we have
	\begin{eqnarray}\label{efficiency}
	\hat{\theta}_n^*-\theta^*=-2 V^{-1}\left\{\frac{1}{n}\sum_{i=1}^n e_i \frac{\partial y^s}{\partial \theta}(x_i,\theta^*)\right\}+o_p(n^{-1/2}).
	\end{eqnarray}
\end{theorem}

We find that the asymptotic representation of $\hat{\theta}_n^*-\theta^*$ agrees with (\ref{benchmark}). This suggests that the proposed projected kernel calibration also achieves the semi-parametric efficiency.

\subsection{Bayesian Interpretation}\label{sec_Bayesian}

From (\ref{PKC}), the projected kernel calibration proceeds by solving a one-step minimization problem, which differs from the $L_2$ calibration in (\ref{L2step1})-(\ref{L2step2}). This formulation allows us to present a Bayesian interpretation for the proposed calibration method.

It follows from the representer's Theorem \cite{wahba1990spline,scholkopf2001generalized} that (\ref{PKC}) is equivalent to the following optimization problem:
\begin{eqnarray}\label{PKmode}
\min_{\theta\in\Theta,\alpha\in\mathbb{R}^n}\frac{1}{n}\sum_{i=1}^n\left(y_i^p-\sum_{j=1}^n\alpha_j K_{\mathcal{G}_\theta}(x_i,x_j)-y^s(x_i,\theta)\right)^2\nonumber \\+\lambda\sum_{i=1}^n\sum_{j=1}^n\alpha_i\alpha_j K_{\mathcal{G}_\theta}(x_i,x_j),
\end{eqnarray}
where $\alpha=(\alpha_1,\ldots,\alpha_n)$.

This expression gives a natural Bayesian interpretation of the projected kernel calibration. Specifically, consider the following Bayesian problem with the likelihood function
\begin{eqnarray}\label{PK1}
&&L_{PK}(\alpha,\theta)\nonumber\\ &\propto&\exp\left\{-\frac{1}{n}\sum_{i=1}^n\left(y_i^p-\sum_{j=1}^n\alpha_j K_{\mathcal{G}_\theta}(x_i,x_j)-y^s(x_i,\theta)\right)^2\right\},
\end{eqnarray}
and the prior distribution
\begin{eqnarray}\label{PK2}
\pi_{PK}(\alpha,\theta)\propto \exp\left\{-\lambda \sum_{i=1}^n\sum_{j=1}^n \alpha_i\alpha_j K_{\mathcal{G}_\theta}(x_i,x_j)\right\}.
\end{eqnarray}
Clearly, (\ref{PKmode}) is the posterior mode of the Bayesian problem (\ref{PK1})-(\ref{PK2}). In view of this connection, we recommend using the posterior distribution given by (\ref{PK1})-(\ref{PK2}) for statistical inference.

In the literature of computer experiments, it is a common practice to use a family of covariance functions $K_\phi$ instead of a fixed kernel $K$, where the hyper-parameter $\phi$ is assumed to have a hyper-prior distribution; see \cite{santner2003design}. A similar extension can be made to the model (\ref{PK1})-(\ref{PK2}) by allowing $K$ to be indexed by a hyper-parameter. Another extension is to incorporate expensive computer experiments. By applying a standard technique in computer experiments, we can model the computer output as a realization of a Gaussian process. The above extensions are straightforward and we omit the details.

Now we make a comparison between the model (\ref{PK1})-(\ref{PK2}) and the model proposed by \cite{kennedy2001bayesian}. For simplicity, we assume that the computer code is cheap and the covariance function $K$ is known. According to the discussions in the preceding paragraph, these simplifications do not affect the general message to be delivered later. This simplified version of the KO model has a likelihood function
\begin{eqnarray}\label{KO1}
&&L_{KO}(\alpha,\theta)\nonumber\\
&\propto& \exp\left\{-\frac{1}{n}\sum_{i=1}^n\left(y_i^p-\sum_{j=1}^n\alpha_j K(x_i,x_j)-y^s(x_i,\theta)\right)^2\right\},
\end{eqnarray}
and prior distribution
\begin{eqnarray}\label{KO2}
\pi_{KO}(\alpha,\theta)\propto \exp\left\{-\lambda \sum_{i=1}^n\sum_{j=1}^n \alpha_i\alpha_j K(x_i,x_j)\right\}.
\end{eqnarray}

By comparing (\ref{PK1})-(\ref{PK2}) and (\ref{KO1})-(\ref{KO2}), it can be seen that the only difference between the KO model and the proposed model is on the use of the kernel: the KO method use the original kernel $K$, while the proposed method uses the projected kernel $K_{\mathcal{G}_\theta}$. \cite{tuo2016calibration} proves that the KO method is inconsistent in calibration. The above discussion shows that this inconsistency problem can be rectified by replacing $K$ with $K_{\mathcal{G}_\theta}$. This modification keeps the major steps of the KO method, and thus most of the computer code for KO method can be reused. Given the wide spread use of the KO method, the proposed method is potentially quite impactful.

We also compare the proposed method with the Bayesian method proposed by \cite{plumlee2016bayesian}, which is based on an orthogonal Gaussian process (OGP) modeling technique \cite{plumlee2016orthogonal}. As addressed in Section \ref{sec_ogp}, the covariance function of an orthogonal Gaussian process is a projected kernel function. Therefore, the likelihood function of the method by \cite{plumlee2016bayesian} is
\begin{eqnarray}\label{OGP1}
&&L_{OGP}(\alpha,\theta)\nonumber\\
&\propto& \exp\left\{-\frac{1}{n}\sum_{i=1}^n\left(y_i^p-\sum_{j=1}^n\alpha_j K(x_i,x_j)-y^s(x_i,\theta)\right)^2\right\},
\end{eqnarray}
and the prior distribution is
\begin{eqnarray}\label{OGP2}
\pi_{OGP}(\alpha,\theta)\propto (\det \mathbf{K}_\theta)^{-1/2}\exp\left\{-\lambda \sum_{i=1}^n\sum_{j=1}^n \alpha_i\alpha_j K(x_i,x_j)\right\},
\end{eqnarray}
where $\mathbf{K}_\theta=(K_{\mathcal{G}_\theta}(x_i,x_j))_{i j}$. The model (\ref{OGP1})-(\ref{OGP2}) is very close to the proposed model (\ref{PK1})-(\ref{PK2}), expect for a determinant factor $(\det \mathbf{K}_\theta)^{-1/2}$ in the prior. The OGP-based model has this determinant factor because it is part of the density function of a multivariate normal distribution. So far we are unclear about the role of the determinant factor from a theoretical point of view.

\section{Discussion}\label{sec_discussion}

In this work, we propose a novel method for the calibration of computer models. The proposed method enjoys three nice properties: consistency, semi-parametric efficiency and it has a Bayesian version. Thus it is a desirable method from both theoretical and practical points of view. The Bayesian version of the proposed method can be regarded as a modification of the widely used KO method. The inconsistency problem of KO method is rectified by this modification.  Another related method is the orthogonal Gaussian process models proposed by \cite{plumlee2016bayesian}. We conjecture that the posterior mode of this model is asymptotically equivalent to that of the proposed method.

At the end of this article, we would like to point out that the proposed projected kernel method can be used in areas beyond computer experiments. In fact, we have proposed a nonparametric regression method, which gives an estimate satisfying certain orthogonality constraints. Since orthogonality is a widely used concept in statistics and machine learning today, the proposed method can be potentially useful in related problems.

\section{Technical Proofs}\label{sec_proofs}

In this section we prove the theorems stated in Section \ref{sec_KernelsProperties} and \ref{sec_PKC}. Some necessary lemmas are also introduced.

\subsection{Proof of Theorem \ref{Thpositivedefinite}}
Because $\tilde{R}$ is continuous, $R$ possesses a nonnegative-valued Fourier transformation. Without loss of generality, we assume $R(0)=1$. Then the spectral theory of stationary processes asserts that there exists a stationary Gaussian process $Z(\cdot)$ over $\mathbb{R}^d$ with correlation function $R(\cdot)$. See, for example, \cite{stein1999interpolation} or any time series textbook for detailed discussions. By normalization, we assume $E(Z(x))=0$ and $E(Z^2(x))=1$.

From (\ref{Kg}) it can be seen that
\begin{eqnarray*}
	K_\mathcal{G}(s,t)=Cov\left(Z(s)-\sum_{i=1}^m e_i(s)\int_\Omega Z(x)e_i(x)d x,\nonumber\right.\\\left. Z(t)-\sum_{i=1}^m e_i(t)\int_\Omega Z(x)e_i(x)d x\right).
\end{eqnarray*}
Thus $K_\mathcal{G}$ is not positive definite if and only if there exists $N\in\mathbb{N}$, distinct points $x_1,\ldots,x_N\in\Omega$ and $\alpha=(\alpha_1,\ldots,\alpha_N)\in\mathbb{R}^N\setminus\{0\}$ such that
\begin{eqnarray}
\sum_{j=1}^N \alpha_j \left(Z(x_j)-\sum_{i=1}^m e_i(x_j)\int_\Omega Z(x)e_i(x)d x\right)=0\label{lineardependence}
\end{eqnarray}
almost surely. Suppose (\ref{lineardependence}) holds for distinct points $x_1,\ldots,x_n$. Then it suffices to prove that $\alpha=0$.

Rearranging (\ref{lineardependence}), we obtain
\begin{eqnarray}
\sum_{j=1}^N\alpha_j Z(x_j)&=&\sum_{j=1}^N\alpha_j\sum_{i=1}^m e_i(x_j)\int_\Omega Z(x)e_i(x)d x\nonumber\\
&=&\int_\Omega Z(x)\left(\sum_{j=1}^N\sum_{i=1}^m \alpha_j e_i(x_i) e_i(x)\right)d x\label{lineardependence2}
\end{eqnarray}
almost surely.
Set $g_0(\cdot)=\sum_{j=1}^N\sum_{i=1}^m \alpha_j e_i(x_i) e_i(\cdot)\in L_2(\Omega)$. For any $x_0\in\mathbb{R}^d$, multiplying by $Z(x_0)$ and taking expectation on both sides of (\ref{lineardependence2}), together with Fubini's Theorem yields
\begin{eqnarray}
\sum_{j=1}^N\alpha_j R(x_0-x_j)=\int_\Omega R(x_0-x) g_0(x) d x,\label{integraleq}
\end{eqnarray}
for all $x_0\in\mathbb{R}^d$.
Define $g_0^E\in L_2(\mathbb{R}^d)$ by
\begin{eqnarray*}
	g_0^E(x)=\begin{cases}
		g_0(x), & \text{for $x\in\Omega,$}\\
		0, & \text{for $x\not\in\Omega.$}
	\end{cases}
\end{eqnarray*}
Then (\ref{integraleq}) becomes
\begin{eqnarray}
\sum_{j=1}^N\alpha_j R(x_0-x_j)=\int_{\mathbb{R}^d} R(x_0-x) g_0^E(x) d x,\label{convolution}
\end{eqnarray}
for all $x_0\in\mathbb{R}^d$. Because $\mathfrak{m}(\Omega)<\infty$, we have $L_2(\Omega)\subset L_1(\Omega)$. Therefore, $g_0\in L_1(\Omega)$ and thus $g_0^E\in L_1(\mathbb{R}^d)$. Taking the Fourier transform with respect to $x_0$ on both sides of (\ref{convolution}) and applying the convolution theorem gives
\begin{eqnarray*}
	\sum_{j=1}^N\alpha e^{-i x^T x_i}\tilde{R}(x)=\tilde{R}(x)\widetilde{g_0^E}(x),
\end{eqnarray*}
for $x\in\mathbb{R}^d$, which, together with $\tilde{R}>0$ almost everywhere, implies
\begin{eqnarray}
\sum_{j=1}^N\alpha_i e^{-i x^T x_i}=\widetilde{g_0^E}(x)\label{fourier}
\end{eqnarray}
for $x$ almost everywhere in $\mathbb{R}^d$. Because $x_i$'s are distinct points, it is easily verified that the left hand side of (\ref{fourier}) is not in $L_2(\mathbb{R}^d)$ unless $\alpha=0$. Noting the fact that the Fourier transform of every integrable function in $L_2(\mathbb{R}^d)$ is also in $L_2(\mathbb{R}^d)$, we obtain $\alpha=0$, which completes the proof.

\subsection{Some Lemmas}\label{sec_lemma}

We introduce some lemmas in this section.
Lemma \ref{th_equality} and Lemma \ref{th_kg} gives some simple but useful formulae of the projected operator and the projected kernel. Their proofs are rather elementary and thus are omitted. Lemma \ref{th_embedding} and \ref{Thmercer} are fundamental results in the theory of native spaces. Lemma \ref{th_embedding} describes some simple embedding relationships in native spaces.
Lemma \ref{Thmercer} gives a full characterization for the native spaces with the help of the eigenvalues and eigenfunctions of $\kappa(\Phi,\cdot)$. We refer to \cite{wendland2005scattered} for the proof of Lemma \ref{th_embedding}. A proof of Lemma \ref{Thmercer} can be found in \cite{schaback1999native}.
Lemma \ref{th_projection} plays an important role in the proof of Theorem \ref{th_norminequality}.

\begin{lemma}\label{th_equality}
	For any $u\in L_2(\Omega\times\Omega),f\in L_2(\Omega)$, the following statements are true.
	
		(i) $\mathcal{P}_\mathcal{G}^{(1)}\mathcal{P}_\mathcal{G}^{(2)}u=\mathcal{P}_\mathcal{G}^{(2)}\mathcal{P}_\mathcal{G}^{(1)}u$.
		
		(ii) $\kappa(\mathcal{P}_\mathcal{G}^{(1)} u,f)=\mathcal{P}_\mathcal{G} \kappa(u,f)$.
		
		(iii) $\kappa(\mathcal{P}_\mathcal{G}^{(2)} u,f)=\kappa(u,\mathcal{P}_\mathcal{G}f)$.
\end{lemma}



\begin{lemma}\label{th_kg}
	Suppose $K(\cdot,\cdot)$ is a positive definite function over $\Omega\times\Omega$, $\mathcal{G}$ is a finite dimensional subspace of $L_2(\Omega)$.
	The following statements for $K_\mathcal{G}$ are true.
	
		(i) For all $f\in L_2$, $\kappa(K_\mathcal{G},f)\in \mathcal{G}^\perp$.
		
		(ii) For all $f\in\mathcal{G}$, $\kappa(K_\mathcal{G},f)=0$.
		
		(iii) For all $f\in\mathcal{G}^\perp$, $\kappa(K_\mathcal{G},f)=\mathcal{P}_\mathcal{G}^\perp\kappa(K,f)$.
		
		(iv) For all $f\in\mathcal{G}^\perp$, $\langle f,\kappa(K,f)\rangle_{L_2(\Omega)}=\langle f,\kappa(K_\mathcal{G},f)\rangle_{L_2(\Omega)}$.
\end{lemma}

\begin{lemma}\label{th_embedding}
	Suppose $\Omega\subset \mathbb{R}^d$ is compact and $\Phi$ is a symmetric positive definite function over $\Omega\times\Omega$. Then the native space $\mathcal{N}_\Phi(\Omega)$ has a continuous linear embedding into $L_2(\Omega)$ satisfiying
	\begin{eqnarray}
	\|f\|_{L_2(\Omega)}\leq C\|f\|_{\mathcal{N}_\Phi(\Omega)},\label{lemmaieq1}
	\end{eqnarray}
	with $C=(\int_\Omega \Phi(x,x)d x)^{1/2}$. Moreover, the integral operator $\kappa(\Phi,\cdot)$ maps $L_2(\Omega)$ continuously into $\mathcal{N}_\Phi(\Omega)$ and satisfies
	\begin{eqnarray}
	\langle f, v\rangle_{L_2(\Omega)}=\langle f,\kappa (\Phi,v)\rangle_{\mathcal{N}_\Phi(\Omega)},\label{lemmaieq2}
	\end{eqnarray}
	for all $f\in\mathcal{N}_\Phi(\Omega)$ and $v\in L_2(\Omega)$.
	The range of $\kappa(\Phi,\cdot)$ is dense in $\mathcal{N}_\Phi(\Omega)$.
\end{lemma}



\begin{lemma}\label{Thmercer}
	Let $\Phi(\cdot,\cdot)$ be a continuous and positive definite function over $\Omega\times\Omega$, where $\Omega$ is a compact subset of $\mathbb{R}^d$. Then there is an orthonormal set $\{\phi_i\}_{i=1}^\infty$ in $L_2(\Omega)$ consisting of the eigenfunctions of $\kappa(\Phi,\cdot)$ such that
	$\kappa(\Phi,\phi_i)=\rho_i\phi_i$
	for eigenvalues $\rho_i> 0$ with $\|\phi_i\|_{L_2(\Omega)}=1$. Then the native space generated by $\Phi$ is embedded into $L_2(\Omega)$ with
	\begin{eqnarray}
	\mathcal{N}_\Phi(\Omega)=\left\{f\in L_2(\Omega):f=\sum_{i=1}^\infty\langle f,\phi_i\rangle_{L_2(\Omega)}\phi_i,\right.\nonumber\\ \left.\text{ with }\sum_{i=1}^{\infty}\frac{1}{\rho_i}|\langle f,\phi_i\rangle_{L_2(\Omega)}|^2<\infty\right\} \label{nativespace}
	\end{eqnarray}
	and the inner product has the representation
	\begin{eqnarray}
	\langle f,g\rangle_{\mathcal{N}_\Phi(\Omega)}=\sum_{i=1}^\infty \frac{1}{\rho_i}\langle f,\phi_i\rangle_{L_2(\Omega)}\langle g,\phi_i\rangle_{L_2(\Omega)},\label{nativeinnerproduct}
	\end{eqnarray}
	for $f,g\in\mathcal{N}_\Phi(\Omega)$.
\end{lemma}

\begin{lemma}\label{th_projection}
	Suppose $\Omega\subset \mathbb{R}^d$ is compact and $K$ is a symmetric positive definite function over $\Omega\times\Omega$. Then for any $f\in\mathcal{N}_K(\Omega),h\in L_2(\Omega)$, we have
	\begin{eqnarray*}
		\|\kappa(K,\langle f,h\rangle_{L_2(\Omega)} f)\|_{\mathcal{N}_K(\Omega)}\leq \|\kappa(K,h)\|_{\mathcal{N}_K(\Omega)} \|f\|_{\mathcal{N}_K(\Omega)}\|\kappa(K,f)\|_{\mathcal{N}_K(\Omega)}.
	\end{eqnarray*}
\end{lemma}

\begin{proof}
	Let $\{\phi_i\}_{i=1}^\infty$ be an orthonormal basis of $L_2(\Omega)$ consisting of the eigenfunctions of $\kappa(K,\cdot)$. Denote the set of the corresponding eigenvalues by $\{\rho_i\}_{i=1}^\infty$. Because $f,h\in L_2(\Omega)$, they admits a unique representation using the basis functions, which is denoted as
	\begin{eqnarray*}
		f=\sum_{i=1}^\infty\eta_i\phi_i,h=\sum_{i=1}^\infty\tau_i\phi_i.
	\end{eqnarray*}
	First, by the definition of the eigenvalues and eigenfunctions we have
	\begin{eqnarray*}
		\kappa(K,\langle f,h\rangle_{L_2(\Omega)} f)&=&\kappa\left(K,\sum_{i=1}^\infty\left(\sum_{k=1}^\infty\tau_k\eta_k\right)\eta_n\phi_n\right).\\
		&=&\left(\sum_{k=1}^\infty\tau_k\eta_k\right)\left(\sum_{i=1}^\infty\rho_n\eta_n\phi_n\right).
	\end{eqnarray*}
	Applying Lemma \ref{Thmercer} yields that
	\begin{eqnarray*}
		&&\|\kappa(K,\langle f,h\rangle_{L_2(\Omega)} f)\|^2_{\mathcal{N}_K(\Omega)}\\
		&=&\left(\sum_{k=1}^\infty\tau_k\eta_k\right)^2 \left(\sum_{i=1}^\infty\rho_n\eta_n^2\right)\\
		&\leq&\left(\sum_{k=1}^\infty \rho_k\tau_k^2\right)\left(\sum_{k=1}^\infty\frac{\eta_k^2}{\rho_k}\right)\left(\sum_{k=1}^\infty \rho_k\eta_k^2\right)\\
		&=&\|\kappa(K,h)\|^2_{\mathcal{N}_K(\Omega)} \|f\|^2_{\mathcal{N}_K(\Omega)}\|\kappa(K,f)\|^2_{\mathcal{N}_K(\Omega)},
	\end{eqnarray*}
	where the inequality follows from the Cauchy-Schwarz inequality. The desired result then follows.
\end{proof}

\subsection{Proof of Theorem \ref{th_norminequality}}


First we show (\ref{normineq}) holds for all $f$ in the range of $\kappa(K,\cdot)$.
For any $h\in L_2(\Omega)$, we have the following identity
\begin{eqnarray}
\kappa(K,\mathcal{P}_\mathcal{G}^\perp h)=\kappa(K,h)-\kappa(K,\mathcal{P}_\mathcal{G} h).\label{kappaidentity}
\end{eqnarray}
By Lemma \ref{th_embedding}, the three terms appeared in (\ref{kappaidentity}) lie in $\mathcal{N}_K(\Omega)$. We take the $\mathcal{N}_K(\Omega)$-norm on both sides of (\ref{kappaidentity}) and use the triangle inequality to find
\begin{eqnarray}
\|\kappa(K,\mathcal{P}_\mathcal{G}^\perp h)\|_{\mathcal{N}_K(\Omega)}\leq \|\kappa(K,h)\|_{\mathcal{N}_K(\Omega)}+ \|\kappa(K,\mathcal{P}_\mathcal{G} h)\|_{\mathcal{N}_K(\Omega)}.\label{normineq1}
\end{eqnarray}
Using Lemma \ref{th_embedding}, the square of the left hand side is
\begin{eqnarray}
&&\|\kappa(K,\mathcal{P}_\mathcal{G}^\perp h)\|_{\mathcal{N}_K(\Omega)}^2=\langle \kappa(K,\mathcal{P}_\mathcal{G}^\perp h),\mathcal{P}_\mathcal{G}^\perp h\rangle_{L_2(\Omega)}\nonumber\\
&=&\langle\kappa(K_\mathcal{G},\mathcal{P}_\mathcal{G}^\perp h),\mathcal{P}_\mathcal{G}^\perp h\rangle_{L_2(\Omega)}
=\langle\kappa(K_\mathcal{G}, h), h\rangle_{L_2(\Omega)}\nonumber\\
&=&\|\kappa(K_\mathcal{G},h)\|^2_{\mathcal{N}_{K_\mathcal{G}}(\Omega)}= \|\mathcal{P}_\mathcal{G}^\perp\kappa(K,h)\|^2_{\mathcal{N}_{K_\mathcal{G}}(\Omega)},\label{normequality}
\end{eqnarray}
where the first equality follows from (\ref{lemmaieq2}) with $\Phi=K$; the second equality follows from (iv) of Lemma \ref{th_kg}; the third equality follows from (i) and (ii) of Lemma \ref{th_kg}; the fourth equality follows from (\ref{lemmaieq2}) with $\Phi=K_\mathcal{G}$; the fifth equality follows from (iii) of Lemma \ref{th_kg}. Combining (\ref{normineq1}) and (\ref{normequality}), we have obtained
\begin{eqnarray}
\|\mathcal{P}_\mathcal{G}^\perp\kappa(K,h)\|_{\mathcal{N}_{K_\mathcal{G}}(\Omega)}\leq \|\kappa(K,h)\|_{\mathcal{N}_K(\Omega)}+ \|\kappa(K,\mathcal{P}_\mathcal{G} h)\|_{\mathcal{N}_K(\Omega)}.\label{normineq2}
\end{eqnarray}
Next we want to show that $\|\kappa(K,\mathcal{P}_\mathcal{G} h)\|_{\mathcal{N}_K(\Omega)}$ can be bounded by a multiple of $\|\kappa(K,h)\|_{\mathcal{N}_K(\Omega)}$. Let $h_0=\mathcal{P}_\mathcal{G} h/\|\mathcal{P}_\mathcal{G} h\|_{L_2(\Omega)}$, we have $\|h_0\|_{L_2(\Omega)}=1$ and $\mathcal{P}_\mathcal{G}h=\langle h,h_0\rangle_{L_2(\Omega)}h_0$. Apply Lemma \ref{th_projection} to arrive at
\begin{eqnarray*}
	&&\|\kappa(K,\mathcal{P}_\mathcal{G} h)\|_{\mathcal{N}_K(\Omega)}= \left\|\kappa\left(K,\langle h,h_0\rangle_{L_2(\Omega)}h_0\right)\right\|_{\mathcal{N}_K(\Omega)}\nonumber\\
	&\leq& \|\kappa(K,h)\|_{\mathcal{N}_K(\Omega)} \|h_0\|_{\mathcal{N}_K(\Omega)}\|\kappa(K,h_0)\|_{\mathcal{N}_K(\Omega)}\nonumber\\
	&\leq& \sup_{g\in\mathcal{G},\|g\|_{L_2(\Omega)}=1}\|g\|_{\mathcal{N}_K(\Omega)}\|\kappa(K,g)\|_{\mathcal{N}_K(\Omega)} \|\kappa(K,h)\|_{\mathcal{N}_K(\Omega)},
\end{eqnarray*}
which, together with (\ref{normineq2}), implies
\begin{eqnarray}
\|\mathcal{P}_\mathcal{G}^\perp\kappa(K,h)\|_{\mathcal{N}_{K_\mathcal{G}}(\Omega)}\leq C_1 \|\kappa(K,h)\|_{\mathcal{N}_K(\Omega)},\label{normineq3}
\end{eqnarray}
for all $h\in L_2(\Omega)$.

By Lemma \ref{th_embedding}, the range of $\kappa(K,\cdot)$ is dense in $\mathcal{N}_K(\Omega)$, i.e., for any $f\in \mathcal{N}_K(\Omega)$, there exists a sequence $e_i\in L_2(\Omega)$, $i=1,2,\ldots$, with
\begin{eqnarray}
\|f-\kappa(K,e_i)\|_{\mathcal{N}_K(\Omega)}\rightarrow 0 \label{convergence}
\end{eqnarray}
as $i\rightarrow\infty$. Noting that (\ref{convergence}) implies
\begin{eqnarray*}
	\|\kappa(K,e_i)-\kappa(K,e_j)\|_{\mathcal{N}_K(\Omega)}\rightarrow 0
\end{eqnarray*}
as $i,j\rightarrow\infty$, which, together with (\ref{normineq3}), implies that
\begin{eqnarray*}
	\|\mathcal{P}_\mathcal{G}^\perp\kappa(K,e_i)-\mathcal{P}_\mathcal{G}^\perp\kappa(K,e_j)\|_{\mathcal{N}_{K_\mathcal{G}}(\Omega)}\rightarrow 0,
\end{eqnarray*}
as $i,j\rightarrow\infty$. This suggests that $\mathcal{P}_\mathcal{G}^\perp\kappa(K,e_i)$ is a Cauchy sequence. Thus the completeness of $\mathcal{N}_{K_\mathcal{G}}(\Omega)$ ensures that $\{\mathcal{P}_\mathcal{G}^\perp\kappa(K,e_i)\}_{i=1}^\infty$ is a convergent sequence in $\mathcal{N}_{K_\mathcal{G}}(\Omega)$. By Lemma \ref{th_embedding}, $\mathcal{N}_{K_\mathcal{G}}(\Omega)$ is continuously embedded into $L_2(\Omega)$, which implies that the limiting function of $\mathcal{P}_\mathcal{G}^\perp\kappa(K,e_i)$ in $\mathcal{N}_{K_\mathcal{G}}(\Omega)$ is $\mathcal{P}_\mathcal{G}^\perp f$. The desired result then follows from the continuity of the native norm. This proves (i).


Take $f\in\mathcal{N}_{K_\mathcal{G}}$. The goal is to show that $\|f\|_{\mathcal{N}_K(\Omega)}$ is bounded above by a multiple of $\|f\|_{\mathcal{N}_{K_\mathcal{G}}(\Omega)}$. As before, we first suppose that $f$ is in the range of $\kappa(K_\mathcal{G},\cdot)$, say, $f=\kappa(K_\mathcal{G},h)$. Using (ii) of Lemma \ref{Kg}, we can assume $h\in\mathcal{G}^\perp$ without loss of generality. By (iii) of Lemma \ref{Kg},
\begin{eqnarray}
\kappa(K_\mathcal{G},h) = \mathcal{P}_\mathcal{G}^\perp\kappa(K,h)= \kappa(K,h)-\mathcal{P}_\mathcal{G}\kappa(K,f).\label{subset}
\end{eqnarray}
Let $h_1=\mathcal{P}_\mathcal{G}\kappa(K,f)/\|\mathcal{P}_\mathcal{G}\kappa(K,f)\|_{L_2(\Omega)}$. Then
\begin{eqnarray}
&&\|\mathcal{P}_\mathcal{G}\kappa(K,f)\|_{\mathcal{N}_K(\Omega)}\nonumber\\
&=&\|\langle \kappa(K,f),h_1\rangle_{L_2(\Omega)}h_1\|_{\mathcal{N}_K(\Omega)}\nonumber\\
&\leq& \|\kappa(K,f)\|_{L_2(\Omega)}\|h_1\|_{L_2(\Omega)}\|h_1\|_{\mathcal{N}_K(\Omega)}\nonumber\\
&\leq& \sup_{g\in\mathcal{G},\|g\|_{L_2(\Omega)}=1}\|g\|_{\mathcal{N}_K(\Omega)}\left(\int_\Omega K(x,x)d x\right)^{1/2}\|\kappa(K,f)\|_{\mathcal{N}_K(\Omega)},\label{normineq4}
\end{eqnarray}
where the first inequality follows from the Cauchy-Schwarz inequality; the second inequality follows from the fact that $h_1\in\mathcal{G}$ and Lemma \ref{th_embedding}.
Now we combine (\ref{subset}) and (\ref{normineq4}) to arrive at
\begin{eqnarray}
&&\|\kappa(K_\mathcal{G},h)\|_{\mathcal{N}_K(\Omega)} \leq C_1^{-1} \|\kappa(K,h)\|_{\mathcal{N}_K(\Omega)}\nonumber\\
&=&C_2 \|\kappa(K,\mathcal{P}^\perp_\mathcal{G}h)\|_{\mathcal{N}_K(\Omega)}= C_2\|\kappa(K_\mathcal{G},h)\|_{\mathcal{N}_{K_\mathcal{G}}(\Omega)},\label{normineq5}
\end{eqnarray}
where the last identity follows from (\ref{normequality}). To prove that (\ref{normineq5}) holds for a general $f\in\mathcal{N}_{K_\mathcal{G}}(\Omega)$ as well, one can apply a continuous-extension argument similar to that disclosed in the previous paragraph. This proves (ii).


\subsection{Proof of Theorem \ref{th_prediction}}

Let $I(g,\theta)=\|g\|_{\mathcal{N}_{K_{\mathcal{G}_\theta}}}$ for $\theta\in\Theta$ and $g\in\mathcal{N}_{K_\mathcal{\theta}}(\Omega)$. By Theorem \ref{th_norminequality}, (\ref{rhomax}) and (\ref{supbound}), we find that the following inequality holds for all $\theta\in\Theta$:
\begin{eqnarray}
C_4\|g\|_{\mathcal{N}_K(\Omega)}\leq I(g,\theta) \leq C_5\|g\|_{\mathcal{N}_K(\Omega)},\label{uniformineq}
\end{eqnarray}
with
\begin{eqnarray*}
	C_4^{-1}&=&1+C_3\left(\int_\Omega K(x,x)d x\right)^{1/2},\\
	C_5&=&1+\rho_{max}C_3^2,
\end{eqnarray*}
where $\rho_{max}$ denotes the maximum eigenvalue of $\kappa(K,\cdot)$.

Because $(\hat{\theta}_n,\hat{\delta}_n)$ minimizes (\ref{PKC}), we have
\begin{eqnarray*}
	&& \frac{1}{n}\sum_{i=1}^n(y_i^p-\hat{\zeta}(x_i))^2+\lambda_n \|\hat{\zeta}(\cdot)-y^s(\cdot,\hat{\theta}_n)\|^2_{\mathcal{N}_{K_{\mathcal{G}_{\hat{\theta}_n}}}(\Omega)}\\
	&\leq& \frac{1}{n}\sum_{i=1}^n(y_i^p-\zeta(x_i))^2+\lambda_n \|\zeta(\cdot)-y^s(\cdot,\theta_n)\|^2_{\mathcal{N}_{K_{\mathcal{G}_\theta}}(\Omega)},
\end{eqnarray*}
which, together with (\ref{uniformineq}) and (\ref{supbound2}), yields
\begin{eqnarray*}
	&& \frac{1}{n}\sum_{i=1}^n(y_i^p-\hat{\zeta}(x_i))^2+\lambda_n C_4 \left(\|\hat{\zeta}(\cdot)\|_{\mathcal{N}_K(\Omega)}^2-C'_3\right)^2\\
	&\leq& \frac{1}{n}\sum_{i=1}^n(y_i^p-\zeta(x_i))^2+\lambda_n C_5\|\zeta(\cdot)-y^s(\cdot,\theta_n)\|^2_{\mathcal{N}_{K_{\mathcal{G}_\theta}}(\Omega)}
\end{eqnarray*}

Define the empirical norm $\|h\|^2_n=n^{-1}\sum_{i=1}^n h^2(x_i)$.
Because $\mathcal{N}_K(\Omega)$ can be continuously embedded into $H^m(\Omega)$, we use the metric entropy of balls of Sobolev spaces \cite{vandegeer2000empirical,tuo2016calibration} and Theorem 5.11 of \cite{vandegeer2000empirical} to obtain the modulus of continuity of the empirical process $v(\zeta')=n^{-1/2}\sum_{i=1}^n e_i(\zeta'-\zeta)$ as
\begin{eqnarray}\label{modulus}
\sup_{\zeta'\in\mathcal{N}_K(\Omega)}\frac{n^{-1}\sum_{i=1}^n e_i(\zeta'-\zeta)}{\|\zeta'-\zeta\|^{1-d/2m}_n\|\zeta'\|^{d/2m}_{\mathcal{N}_K(\Omega)}}=O_p(n^{-1/2}).
\end{eqnarray}

The remainder of the proof is standard. Using the arguments similar to the proof of Theorem 10.2 of \cite{vandegeer2000empirical}, we obtain (\ref{Hmrate}) and
\begin{eqnarray}\label{empiricalrate}
\|\zeta-\hat{\zeta}\|_n =O_p(n^{-\frac{2 m}{2m+d}}).
\end{eqnarray}
Now invoking the condition that $x_i$'s follow the uniform distribution over $\Omega$, we apply Lemma 5.16 of \cite{vandegeer2000empirical} to conclude an asymptotic equivalence relation between the $L_2$ and the empirical norm as:
\begin{eqnarray}\label{equivalence}
\limsup_{n\rightarrow\infty}\mathbf{P}\left(\sup_{\substack{\|h\|_{\mathcal{N}_K(\Omega)}=O_p(1)\\ \|h\|_{L_2(\Omega)}>\tau n^{-\frac{m}{2m+d}}/\eta}}\left|\frac{\|h\|_n}{\|h\|_{L_2(\Omega)}}-1\right|\geq \eta\right)=0.
\end{eqnarray}
Combining (\ref{Hmrate}), (\ref{empiricalrate}) and (\ref{equivalence}), we prove (\ref{L2rate}).


\subsection{Proof of Theorem \ref{th_calibration}}

Let $f(\theta,\delta)=\frac{1}{n}\sum_{i=1}^n (y_i^p-\delta(x_i)-y^s(x_i,\theta))^2+\lambda_n\|\delta\|^2_{\mathcal{N}_{K_{\mathcal{G}_\theta}}(\Omega)}$ and $\delta^*(x)=\zeta(x)-y^s(x,\theta^*)$. We prove the desired results by showing that
\begin{eqnarray}
f(\theta^*,\delta^*) \leq \inf_{\|\theta-\theta^*\|=c n^{-\frac{m}{2m+d}},\delta} f(\theta,\delta),\label{localmin}
\end{eqnarray}
for sufficiently large $n$ and some constant $c>0$ to be specified later, where $\|\cdot\|$ denotes the usual Euclidean distance.

First we observe that for fixed $\theta\in\Theta$,
\begin{eqnarray*}
	&&\inf_{\delta\in\mathcal{N}_{K_{\mathcal{G}_\theta}}}\|\zeta(\cdot)-y^s(\cdot,\theta)-\delta(\cdot)\|^2_{L_2(\Omega)}\\
	&\geq & \inf_{\delta\perp \mathcal{G}_\theta}\|\zeta(\cdot)-y^s(\cdot,\theta)-\delta(\cdot)\|^2_{L_2(\Omega)}\\
	&=&\|\mathcal{P}_{\mathcal{G}_\theta}(\zeta(\cdot)-y^s(\cdot,\theta))\|^2_{L_2(\Omega)}\\
	&=&\sum_{i=1}^q \left\langle \zeta(\cdot)-y^s(\cdot,\theta), \frac{\partial y^s}{\partial \theta_i}(\cdot,\theta)\right\rangle^2_{L_2(\Omega)} \Big/ \left\| \frac{\partial y^s}{\partial \theta_i}(\cdot,\theta)\right\|^2_{L_2(\Omega)}\\
	&\geq& \left\|\int_\Omega (\zeta(x)-y^s(x,\theta))\frac{\partial y^s}{\partial \theta}(x,\theta) d x \right\|^2\Big/ \max_{1\leq i\leq q}\left\| \frac{\partial y^s}{\partial \theta_i}(\cdot,\theta)\right\|^2_{L_2(\Omega)}\\
	&=&\left\|\frac{1}{2}\int_\Omega \frac{\partial }{\partial \theta}(\zeta(x)-y^s(x,\theta))^2 d x \right\|^2\Big/ \max_{1\leq i\leq q}\left\| \frac{\partial y^s}{\partial \theta_i}(\cdot,\theta)\right\|^2_{L_2(\Omega)}.
\end{eqnarray*}
Therefore, using (\ref{supbound3}), Lebesgue's dominated convergence theorem and the positive definiteness of $V$, we can find constants $c_1,c_2>0$, such that for $\|\theta-\theta^*\|\leq c_1$,
\begin{eqnarray}
\inf_{\delta\in\mathcal{N}_{K_{\mathcal{G}_\theta}}}\|\zeta(\cdot)-y^s(\cdot,\theta)-\delta(\cdot)\|_{L_2(\Omega)} \geq c_2 \|\theta-\theta^*\|.\label{normlowerbd}
\end{eqnarray}

Now suppose (\ref{localmin}) is false. Then there exists $\tilde{\theta}$ with $\|\theta^*-\tilde{\theta}\|=c n^{-\frac{m}{2m+d}}$ so that
\begin{eqnarray}\label{basiceq1}
&&\frac{1}{n}\sum_{i=1}^n e_i^2+\lambda_n\|\delta^*\|^2_{\mathcal{N}_{K_{\mathcal{G}_{\theta^*}}}(\Omega)}\nonumber\\
&>&\frac{1}{n}\sum_{i=1}^n(e_i+\zeta(x_i)-\tilde{\delta}(x_i)-y^s(x_i,\tilde{\theta}))^2+\lambda_n\|\tilde{\delta}\|^2_{\mathcal{N}_{K_{\mathcal{G}_{\tilde{\theta}}}}(\Omega)},
\end{eqnarray}
which is equivalent to
\begin{eqnarray}\label{basiceq2}
&&\frac{2}{n}\sum_{i=1}^n e_i(\tilde{\delta}(x_i)+y^s(x_i,\tilde{\theta})-\zeta(x_i))+\lambda_n \|\delta^*\|^2_{\mathcal{N}_{K_{\mathcal{G}_{\theta^*}}}(\Omega)}\nonumber\\
&>&\|\tilde{\delta}(\cdot)+y^s(\cdot,\tilde{\theta})-\zeta(\cdot)\|_n^2+\lambda_n\|\tilde{\delta}\|^2_{\mathcal{N}_{K_{\mathcal{G}_{\tilde{\theta}}}}(\Omega)}.
\end{eqnarray}
From (\ref{basiceq1}) we have
\begin{eqnarray}\label{normupperbd}
\lambda_n\|\tilde{\delta}\|^2_{\mathcal{N}_{K_{\mathcal{G}_{\tilde{\theta}}}}(\Omega)} < \frac{1}{n}\sum_{i=1}^n e_i^2+\lambda_n\|\delta^*\|^2_{\mathcal{N}_{K_{\mathcal{G}_{\theta^*}}}(\Omega)}=O_p(\lambda_n),
\end{eqnarray}
which implies $\|\tilde{\delta}\|_{\mathcal{N}_{K_{\mathcal{G}_{\tilde{\theta}}}}(\Omega)}=O_p(1)$. Then analogous to (\ref{modulus}) we have
\begin{eqnarray}\label{modulus2}
&& \frac{2}{n}\sum_{i=1}^n e_i(\tilde{\delta}(x_i)+y^s(x_i,\tilde{\theta})-\zeta(x_i)) \\ &=&O_p(n^{-1/2})\|\tilde{\delta}(\cdot)+y^s(\cdot,\tilde{\theta})-\zeta(\cdot)\|_n^{1-d/2m}.
\end{eqnarray}
Combining (\ref{basiceq2})-(\ref{modulus2}), we arrive at
\begin{eqnarray*}
	O_p(n^{-1/2})\|\tilde{\delta}(\cdot)+y^s(\cdot,\tilde{\theta})-\zeta(\cdot)\|_n^{1-d/2m} >\|\tilde{\delta}(\cdot)+y^s(\cdot,\tilde{\theta})-\zeta(\cdot)\|_n^2,
\end{eqnarray*}
which gives $\|\tilde{\delta}(\cdot)+y^s(\cdot,\tilde{\theta})-\zeta(\cdot)\|_n=O_p(n^{-\frac{m}{2m+d}})$. This implies that there exists a constant $K_0$ such that the event $\{\|\tilde{\delta}(\cdot)+y^s(\cdot,\tilde{\theta})-\zeta(\cdot)\|_n\leq K_0 n^{-\frac{m}{2m+d}}\}$ occurs with probability tending to one.

On the other hand, choose $c=\max\{2K_0,\tau\}/c_2$, where $\tau$ is defined in (\ref{equivalence}). Choose $n$ sufficiently large so that $c n^{-\frac{m}{2m+d}}<c_1$. Then by (\ref{normlowerbd}) we have
\begin{eqnarray*}
	\|\tilde{\delta}(\cdot)+y^s(\cdot,\tilde{\theta})-\zeta(\cdot)\|_{L_2(\Omega)}\geq c_2 c n^{-\frac{m}{2m+d}},
\end{eqnarray*}
which, together with (\ref{normlowerbd}) and (\ref{equivalence}), yields
\begin{eqnarray}
\|\tilde{\delta}(\cdot)+y^s(\cdot,\tilde{\theta})-\zeta(\cdot)\|_n\geq(1+o_p(1))2K_0 n^{-\frac{m}{2m+d}}.
\end{eqnarray}
This leads to a contradiction.

\subsection{Calculus for Projected Kernels}

In this section, we introduce some calculus results which help prove the asymptotic normality of the projected kernel calibration.

Let $\mathcal{G}_\theta=\operatorname{span}\{g_{1\theta},g_{2\theta},\ldots,g_{q\theta}\}$ for $\theta\in\Theta$, where $g_{1\theta},\ldots,g_{q\theta}\in L_2(\Omega)$ are linearly independent and are differentiable with respect to $\theta$. Define the ``distance matrix'' $E_\theta=(\langle g_{i\theta},g_{i\theta}\rangle_{L_2(\Omega)})_{ij}$ and $\mathbf{g}_\theta=(g_{1\theta},\ldots,g_{q\theta})^T$.

Fix $f\in L_2(\Omega)$ and let $b_\theta=(\langle f,g_{1\theta}\rangle_{L_2(\Omega)},\ldots,\langle f,g_{q\theta}\rangle_{L_2(\Omega)})^T$. It follows from basic linear algebra that
\begin{eqnarray}\label{projection}
\mathcal{P}_{\mathcal{G}_\theta}f=b^T_\theta E_\theta^{-1}\mathbf{g}_\theta.
\end{eqnarray}

The goal of this section is to bound the derivative of the norm in the native spaces generated by projected kernels. 

\begin{lemma}\label{le_dev}
	Suppose $\sup_{1\leq j\leq q,x\in\Omega}|\frac{\partial g_j}{\partial \theta_i}(x)|<\infty$.
	Then, for any $f\in\mathcal{N}_K(\Omega)$,
	\begin{eqnarray*}
		\frac{\partial}{\partial\theta_i}\|\mathcal{P}^\perp_{\mathcal{G}_\theta}f\|^2_{\mathcal{N}_{K_{\mathcal{G}_\theta}}(\Omega)}  \leq 2\|f\|^2_{\mathcal{N}_K(\Omega)} \|\mathbf{v}_1\|\|\mathbf{v}_2\|/\lambda_{min}(E_\theta),
	\end{eqnarray*}
	where
	\begin{eqnarray*}
		\mathbf{v}_1 &=& \left(\|\kappa(K,g_{1\theta})\|_{\mathcal{N}_K(\Omega)},\ldots,\|\kappa(K,g_{q\theta})\|_{\mathcal{N}_K(\Omega)}\right)^T, \\
		\mathbf{v}_2 &=& \left(\left\|\frac{\partial g_{1\theta}}{\partial \theta_i}\right\|_{\mathcal{N}_K(\Omega)},\ldots,\left\|\frac{\partial g_{q\theta}}{\partial \theta_i}\right\|_{\mathcal{N}_K(\Omega)}\right)^T,
	\end{eqnarray*}
	and $\lambda_{min}(E_\theta)$ denotes the minimum eigenvalue of $E_\theta$.
\end{lemma}

\begin{proof}
	Following a continuous extension argument similar to that used in Theorem \ref{th_norminequality}, it suffices to prove that the desired inequality holds if $f$ lies in the range of $\kappa(K,\cdot)$. Suppose $f=\kappa(K,h)$ for $h\in L_2(\Omega)$. By identity (\ref{normequality}),
	\begin{eqnarray}\label{integral}
	\|\mathcal{P}^\perp_{\mathcal{G}_\theta}\kappa(K,h)\|^2_{\mathcal{N}_{K_{\mathcal{G}_\theta}}(\Omega)} &=& \int_{\Omega\times\Omega} h(x)K_{\mathcal{G}_\theta}(x,y)h(y)d x d y\\
	&=& \int_{\Omega\times\Omega} (\mathcal{P}^\perp_{\mathcal{G}_\theta}h)(x)K_{\mathcal{G}_\theta}(x,y)(\mathcal{P}^\perp_{\mathcal{G}_\theta}h)(y)d x d y.\nonumber
	\end{eqnarray}
	Hence, without loss of generality, we assume $h\in \mathcal{G}_\theta^{\perp}$.
	From the definition of $K_{\mathcal{G}_\theta}$ and the uniform boundedness of $\frac{\partial g_j}{\partial \theta_i}$, the conditions of Lebesgue's dominated convergence theorem are fulfilled, allowing us to calculate the derivative of (\ref{integral}) by interchanging differentiation and integration, which gives
	\begin{eqnarray*}
		\frac{\partial}{\partial\theta_i}\|\mathcal{P}^\perp_{\mathcal{G}_\theta}\kappa(K,h)\|^2_{\mathcal{N}_{K_{\mathcal{G}_\theta}}(\Omega)} = \int_{\Omega\times\Omega} h(x)\frac{\partial K_{\mathcal{G}_\theta}}{\partial\theta_i}(x,y)h(y)d x d y. 
	\end{eqnarray*}
	Use the definition of $K_{\mathcal{G}_\theta}$ to obtain
	\begin{eqnarray*}
		\frac{\partial}{\partial\theta_i}K_{\mathcal{G}_\theta} = \frac{\partial}{\partial\theta_i}\left\{-\mathcal{P}^{(1)}_{\mathcal{G}_\theta}K-\mathcal{P}^{(2)}_{\mathcal{G}_\theta}K+ \mathcal{P}^{(1)}_{\mathcal{G}_\theta}\mathcal{P}^{(2)}_{\mathcal{G}_\theta}K\right\}.
	\end{eqnarray*}
	
	Now we bound these three terms separately. Using (\ref{projection}), we get
	\begin{eqnarray*}
		\mathcal{P}^{(1)}_{\mathcal{G}_\theta}K(x,y) = b^T_\theta(y) E_\theta^{-1}\mathbf{g}_\theta(x),
	\end{eqnarray*}
	with $b^T_\theta(y)=(\int_{\Omega} K(s,y)g_1(s) d s,\ldots, \int_{\Omega} K(s,y)g_q(s) d s)$. Applying Leibniz's rule, we obtain
	\begin{eqnarray}\label{leibniz1}
	\frac{\partial}{\partial \theta_i}\mathcal{P}^{(1)}_{\mathcal{G}_\theta}K(x,y) = \frac{\partial(b^T_\theta E_\theta^{-1})}{\partial \theta_i}(y)\mathbf{g}_\theta(x)+ b^T_\theta(y) E_\theta^{-1}\frac{\partial\mathbf{g}_\theta}{\partial\theta_i}(x).
	\end{eqnarray}
	Noting that $h\perp \mathbf{g}_\theta$, the first term of
	$\int_{\Omega\times\Omega}h(x)\frac{\partial\mathcal{P}^{(1)}_{\mathcal{G}_\theta}K}{\partial \theta_i}(x,y) h(y) d x d y$
	corresponding to (\ref{leibniz1}) vanishes and we use dominated convergence theorem to find
	\begin{eqnarray}\label{derivativebound}
	&&\left|\int_{\Omega\times\Omega}h(x)\frac{\partial\mathcal{P}^{(1)}_{\mathcal{G}_\theta}K}{\partial \theta_i}(x,y) h(y) d x d y\right|\nonumber \\
	&=& \left|\left(\int_{\Omega} h(y) K(s,y)\mathbf{g}^T_\theta(s) d s d y\right) E_\theta^{-1} \int_{\Omega} \frac{\partial \mathbf{g}_\theta}{\partial\theta_i}(x) h(x) d x \right| \nonumber \\
	&=&\left| \langle \kappa(K,h),\kappa(K,\mathbf{g}^T_\theta)\rangle_{\mathcal{N}_K(\Omega)} E_\theta^{-1} \left\langle \frac{\partial \mathbf{g}_\theta}{\partial \theta_i}, \kappa(K,h)\right\rangle_{\mathcal{N}_K(\Omega)}\right| \nonumber \\
	&\leq& \|\kappa(K,h)\|^2_{\mathcal{N}_K(\Omega)}\|\mathbf{v}_1\|\|\mathbf{v}_2\|/\lambda_{min}(E_\theta),
	\end{eqnarray}
	where the second equality follows from (\ref{nativeinnerproduct}); the inequality follows from Cauchy-Schwarz inequality and the standard theory of quadratic forms.
	
	For a similar reason, $|\int_{\Omega\times\Omega}h(x)\frac{\partial\mathcal{P}^{(2)}_{\mathcal{G}_\theta}K}{\partial \theta_i}(x,y) h(y) d x d y|$ is bounded above by the right hand side of (\ref{derivativebound}) as well.
	
	Now we use (\ref{projection}) twice to find that $\mathcal{P}_{\mathcal{G}_\theta}^{(1)}\mathcal{P}_{\mathcal{G}_\theta}^{(2)}K(x,y)$ has the form
	\begin{eqnarray*}
		\sum_{1\leq j,k\leq q} c_{j k \theta} g_{i\theta}(x)g_{j\theta}(y),
	\end{eqnarray*}
	for some $c_{j k \theta}$ independent of $x$ and $y$. Then, $\frac{\partial}{\partial\theta_i}\mathcal{P}_{\mathcal{G}_\theta}^{(1)}\mathcal{P}_{\mathcal{G}_\theta}^{(2)}K(x,y)$ is given by
	\begin{eqnarray}\label{devp1p2}
	\sum_{1\leq j,k\leq q}\left\{\frac{\partial c_{j k \theta}}{\partial\theta_i}g_{i\theta}(x)g_{j\theta}(y)+c_{j k \theta}\frac{\partial g_{i\theta}}{\partial\theta_i}(x)g_{j\theta}(y)\nonumber\right.\\\left.+ c_{j k \theta}g_{i\theta}(x)\frac{\partial g_{j\theta}}{\partial\theta_i}(y)\right\}.
	\end{eqnarray}
	Because each term in (\ref{devp1p2}) has factor either $g_{i\theta}(x)$ or $g_{j\theta}(y)$, we use the condition $h\perp \mathbf{g}_\theta$ to conclude that
	\begin{eqnarray*}
		\int_{\Omega\times\Omega} h(x)\frac{\partial}{\partial\theta_i}\mathcal{P}_{\mathcal{G}_\theta}^{(1)}\mathcal{P}_{\mathcal{G}_\theta}^{(2)}K(x,y) h(y) =0.
	\end{eqnarray*}
	The proof is then completed.
\end{proof}

\subsection{Proof of Theorem \ref{th_normality}}
Define
\begin{eqnarray*}
	l(\theta,\delta) = \frac{1}{n}\sum_{i=1}^n (y_i^p-\mathcal{P}^\perp_{\mathcal{G}_\theta}\delta(x_i)-y^s(x_i,\theta))^2+ \lambda_n\|\mathcal{P}^\perp_{\mathcal{G}_\theta}\delta\|^2_{\mathcal{N}_{K_{\mathcal{G}_\theta}}(\Omega)}.
\end{eqnarray*}
It follows from Theorem \ref{th_norminequality} that $l(\theta,\delta)$ is finite for all $\theta\in\Theta$ and $\delta\in\mathcal{N}_K(\Omega)$. It can be seen from (\ref{projection}) and Lemma \ref{le_dev} that $l$ is differentiable with respect to $\theta$ for fixed $\delta$.
Because $\mathcal{P}^\perp_{\mathcal{G}_{\hat{\theta}^*_n}}\hat{\delta}_n^*=\hat{\delta}_n$, $(\hat{\theta}^*_n,\hat{\delta}_n)$ minimizes $l$.  Then we have
\begin{eqnarray}\label{variational}
\frac{\partial l(\cdot,\hat{\delta}_n)}{\partial \theta}\Big|_{\theta=\hat{\theta}^*_n}=0. 
\end{eqnarray}
Because $\mathcal{P}^\perp_{\mathcal{G}_\theta}\hat{\delta}_n^*=\hat{\delta}_n^*-\mathcal{P}_{\mathcal{G}_\theta}\hat{\delta}_n^*$, we use (\ref{projection}) and Leibniz's rule to derive
\begin{eqnarray}\label{devperp}
\frac{\partial}{\partial\theta}\mathcal{P}^\perp_{\mathcal{G}_\theta}\hat{\delta}_n^*\Big|_{\theta=\hat{\theta}_n^*}=-\frac{\partial b^T_\theta}{\partial\theta}\Big|_{\theta=\hat{\theta}_n^*}D_{\hat{\theta}_n^*}^{-1}\frac{\partial y^s}{\partial\theta}(\cdot,\hat{\theta}_n^*)-b^T_{\hat{\theta}_n^*}\frac{\partial (D_\theta^{-1} \frac{\partial y^s}{\partial\theta})}{\partial\theta}\Big|_{\theta=\hat{\theta}_n^*},
\end{eqnarray}
where $b_\theta=\langle \hat{\delta}_n^*(\cdot), \frac{\partial y^s(\cdot,\theta)}{\partial\theta}\rangle_{L_2(\Omega)}$. Because $\hat{\theta}^*_n\perp \mathcal{G}_{\hat{\theta}^*_n}$, we have $b_{\hat{\theta}_n^*}=0$ and thus the second term of (\ref{devperp}) vanishes. Using Theorem \ref{th_norminequality}, Lemma \ref{le_dev}, (\ref{Hmrate}), (\ref{supbound3}) and (\ref{supbound4}), we conclude that
\begin{eqnarray}\label{op}
\frac{\partial}{\partial\theta} \lambda_n\|\mathcal{P}^\perp_{\mathcal{G}_\theta}\hat{\theta}^*_n\|^2_{\mathcal{N}_{K_{\mathcal{G}_\theta}}(\Omega)} =O_p(\lambda_n)=o_p(n^{-1/2}).
\end{eqnarray}
Combining (\ref{variational}), (\ref{devperp}) and (\ref{op}) we obtain
\begin{eqnarray}\label{firsteq}
&& \frac{1}{n}\sum_{i=1}^n (y_i^p-\hat{\delta}_n^*(x_i)-y^s(x_i,\hat{\theta}_n^*))\left(I-\frac{\partial b^T_\theta}{\partial\theta}\Big|_{\theta=\hat{\theta}_n^*}D_{\hat{\theta}^*_n}^{-1}\right)\frac{\partial y^s}{\partial\theta}(x_i,\hat{\theta}_n^*)\nonumber\\
&=&\frac{1}{n}\sum_{i=1}^n (y_i^p-\hat{\delta}_n^*(x_i)-y^s(x_i,\hat{\theta}_n^*))\left(\left(D_{\hat{\theta}^*_n}-\frac{\partial b^T_\theta}{\partial\theta}\Big|_{\theta=\hat{\theta}_n^*}\right)D_{\hat{\theta}^*_n}^{-1}\right)\frac{\partial y^s}{\partial\theta}(x_i,\hat{\theta}_n^*),\nonumber\\
&=&o_p(n^{-1/2}),
\end{eqnarray}
where $I$ denotes the identity matrix. Direct calculation shows
\begin{eqnarray*}
	2 D_{\hat{\theta}^*_n}-2 \frac{\partial b^T_\theta}{\partial\theta}\Big|_{\theta=\hat{\theta}_n^*}= \int_{\Omega} \frac{\partial^2}{\partial\theta^T\partial\theta}(\hat{\zeta}_n^*(x)-y^s(x,\hat{\theta}_n^*))^2 d x,
\end{eqnarray*}
where $\hat{\zeta}_n^*(\cdot)=\hat{\delta}_n^*(\cdot)+y^s(\cdot,\hat{\theta}_n^*)$. Therefore, by (\ref{L2rate}) and the consistency of $\hat{\theta}^*_n$, $\int_{\Omega} \frac{\partial^2}{\partial\theta^T\partial\theta}(\hat{\zeta}_n^*(x)-y^s(x,\hat{\theta}_n^*))^2 d x$ tends to $I^*$ in probability. Hence by (\ref{Istar}) and (\ref{eigenbound}), the matrix $(D_{\hat{\theta}^*_n}-\frac{\partial b^T_\theta}{\partial\theta}|_{\theta=\hat{\theta}_n^*})D_{\hat{\theta}^*_n}^{-1}$ is invertible with probability tending to one, which, together with (\ref{firsteq}), implies
\begin{eqnarray}
o_p(n^{-1/2})&=&  \frac{1}{n}\sum_{i=1}^n (y_i^p-\hat{\delta}_n^*(x_i)-y^s(x_i,\hat{\theta}_n^*))\frac{\partial y^s}{\partial\theta}(x_i,\hat{\theta}_n^*)\nonumber\\
&=& \frac{1}{n}\sum_{i=1}^n (\zeta(x_i)-\hat{\delta}_n^*(x_i)-y^s(x_i,\hat{\theta}_n^*))\frac{\partial y^s}{\partial\theta}(x_i,\hat{\theta}_n^*)\nonumber\\ &&+\frac{1}{n}\sum_{i=1}^n e_i \frac{\partial y^s}{\partial\theta}(x_i,\hat{\theta}_n^*)\nonumber\\
&=:&I_1+I_2.\label{I1I2}
\end{eqnarray}
Because $\|\hat{\delta}_n^*\|_{\mathcal{N}_K(\Omega)}=O_p(1)$, $\hat{\delta}_n^*$ lies in a Donsker set with probability tending to one; see \cite{vandegeer2000empirical,tuo2014efficient}. Then we invoke the asymptotic equicontinuity of the empirical processes \cite{vandegeer2000empirical} and the conditions $\|\zeta(\cdot)-\hat{\delta}_n^*(\cdot)-y^s(\cdot,\hat{\theta}_n^*)\|_{L_2(\Omega)}=o_p(1)$ to conclude
\begin{eqnarray}\label{I1}
I_1 &=& \int_{\Omega} (\zeta(x)-\hat{\delta}_n^*-y^s(x,\hat{\theta}_n^*))\frac{\partial y^s}{\partial\theta}(x,\hat{\theta}_n^*)d x +o_p(n^{-1/2})\nonumber\\
&=& \int_{\Omega} (\zeta(x)-y^s(x,\hat{\theta}_n^*))\frac{\partial y^s}{\partial\theta}(x,\hat{\theta}_n^*)d x+o_p(n^{-1/2})\nonumber\\
&=& \frac{1}{2}\int_{\Omega}\frac{\partial}{\partial \theta^T\partial\theta}(\zeta(x)-y^s(x,\tilde{\theta}_n^*)) d x (\hat{\theta}_n^*-\theta^*)+o_p(n^{-1/2}),
\end{eqnarray}
with some $\tilde{\theta}_n^*$ lying between $\hat{\theta}_n^*$ and $\theta^*$,
where the second equality follows from the fact that $\hat{\delta}_n^*\perp \mathcal{G}_{\hat{\theta}_n^*}$; the last equality follows from Taylor expansion and the fact that $\int_{\Omega} (\zeta(x)-y^s(x,\theta^*))\frac{\partial y^s}{\partial\theta}(x,\theta^*)d x=0$. A similar asymptotic equicontinuity argument leads to an approximation to $I_2$:
\begin{eqnarray}\label{I2}
I_2 =\frac{1}{n}\sum_{i=1}^n e_i \frac{\partial y^s}{\partial\theta}(x_i,\theta^*)+o_p(n^{-1/2}).
\end{eqnarray}
The desired result follows from combining (\ref{I1I2}), (\ref{I1}) and (\ref{I2}).

\appendix

\bibliographystyle{siamplain}
\bibliography{calibration}

\begin{thebibliography}{10}

\bibitem{bayarri2007framework}
{\sc M.~Bayarri, J.~Berger, R.~Paulo, J.~Sacks, J.~Cafeo, J.~Cavendish, C.~Lin,
  and J.~Tu}, {\em A framework for validation of computer models},
  Technometrics, 49 (2007), pp.~138--154.

\bibitem{goldstein2004probabilistic}
{\sc M.~Goldstein and J.~Rougier}, {\em Probabilistic formulations for
  transferring inferences from mathematical models to physical systems}, SIAM
  Journal on Scientific Computing, 26 (2004), pp.~467--487.

\bibitem{gramacy2015calibrating}
{\sc R.~B. Gramacy, D.~Bingham, J.~P. Holloway, M.~J. Grosskopf, C.~C. Kuranz,
  E.~Rutter, M.~Trantham, R.~P. Drake, et~al.}, {\em Calibrating a large
  computer experiment simulating radiative shock hydrodynamics}, The Annals of
  Applied Statistics, 9 (2015), pp.~1141--1168.

\bibitem{higdon2008computer}
{\sc D.~Higdon, J.~Gattiker, B.~Williams, and M.~Rightley}, {\em Computer model
  calibration using high-dimensional output}, Journal of the American
  Statistical Association, 103 (2008), pp.~570--583.

\bibitem{higdon2004combining}
{\sc D.~Higdon, M.~Kennedy, J.~Cavendish, J.~Cafeo, and R.~Ryne}, {\em
  Combining field data and computer simulations for calibration and
  prediction}, SIAM Journal of Scientific Computing, 26 (2004), pp.~448--466.

\bibitem{joseph2014engineering}
{\sc V.~R. Joseph and H.~Yan}, {\em Engineering-driven statistical adjustment
  and calibration}, Technometrics, 57 (2015), pp.~257--267.

\bibitem{kennedy2001bayesian}
{\sc M.~Kennedy and A.~O'Hagan}, {\em Bayesian calibration of computer models},
  Journal of the Royal Statistical Society: Series B, 63 (2001), pp.~425--464.

\bibitem{plumlee2016bayesian}
{\sc M.~Plumlee}, {\em Bayesian calibration of inexact computer models},
  Journal of the American Statistical Association,  (2016).

\bibitem{plumlee2016orthogonal}
{\sc M.~Plumlee and V.~R. Joseph}, {\em Orthogonal gaussian process models},
  tech. report, University of Michigan and Georgia Institute of Technology,
  2016.

\bibitem{santner2003design}
{\sc T.~Santner, B.~Williams, and W.~Notz}, {\em The Design and Analysis of
  Computer Experiments}, Springer Verlag, 2003.

\bibitem{schaback1999native}
{\sc R.~Schaback}, {\em Native hilbert spaces for radial basis functions i}, in
  New Developments in Approximation Theory, Springer, 1999, pp.~255--282.

\bibitem{scholkopf2001generalized}
{\sc B.~Sch{\"o}lkopf, R.~Herbrich, and A.~J. Smola}, {\em A generalized
  representer theorem}, in International Conference on Computational Learning
  Theory, Springer, 2001, pp.~416--426.

\bibitem{stein1999interpolation}
{\sc M.~Stein}, {\em Interpolation of Spatial Data: Some Theory for Kriging},
  Springer Verlag, 1999.

\bibitem{stone1982optimal}
{\sc C.~J. Stone}, {\em Optimal global rates of convergence for nonparametric
  regression}, The annals of statistics,  (1982), pp.~1040--1053.

\bibitem{tuo2014efficient}
{\sc R.~Tuo and C.~F.~J. Wu}, {\em Efficient calibration for imperfect computer
  models.}, The Annals of Statistics, 43 (2015), pp.~2331--2352.

\bibitem{tuo2016calibration}
{\sc R.~Tuo and C.~F.~J. Wu}, {\em A theoretical framework for calibration in
  computer models: parametrization, estimation and convergence properties},
  SIAM/ASA Journal on Uncertainty Quantification, 4 (2016), pp.~767--795.

\bibitem{vandegeer2000empirical}
{\sc S.~A. van~der Geer}, {\em Empirical Processes in M-estimation}, vol.~6,
  Cambridge university press, 2000.

\bibitem{wahba1990spline}
{\sc G.~Wahba}, {\em Spline Models for Observational Data}, vol.~59, Society
  for Industrial Mathematics, 1990.

\bibitem{wendland2005scattered}
{\sc H.~Wendland}, {\em Scattered Data Approximation}, Cambridge University
  Press, 2005.

\end{thebibliography}

\end{document}